\documentclass[11pt]{llncs}
\usepackage{fullpage}
\usepackage{amsmath}
\usepackage{xspace}
\usepackage{amssymb}

\newtheorem{fact}{Fact}    
\newtheorem{result}{Result}    
\newenvironment{proofof}[1]{\noindent{\bf Proof of #1:}}{\qed}

\newcommand{\defeq}{\stackrel{\mathsf{def}}{=}}
\newcommand{\ket}[1]{| #1 \rangle}
\newcommand{\bra}[1]{\langle #1 |}
\newcommand{\ketbra}[1]{| #1 \rangle \langle #1 |}
\newcommand{\braket}[2]{\langle #1 | #2 \rangle}

\newcommand{\trnorm}[1]{\left\| #1 \right\|_{{\mathrm{tr}}}}
\newcommand{\totvar}[1]{\left\| #1 \right\|_1}

\newcommand{\C}{{\mathbb C}}

\newcommand{\EQ}{{\mathrm{EQ}}}

\newcommand{\E}{{\mathbb{E}}}
\newcommand{\av}{{\mathbb{E}}}

\newcommand{\Tr}{{\mathsf{Tr}}}
\newcommand{\Good}{{\mathsf{Good}}}

\newcommand{\cH}{{\cal H}}
\newcommand{\cK}{{\cal K}}
\newcommand{\cM}{{\cal M}}
\newcommand{\tcM}{\tilde{{\cal M}}}

\newcommand{\cP}{{\cal P}}
\newcommand{\tcP}{{\tilde{\cal P}}}
\newcommand{\cX}{{\cal X}}
\newcommand{\cY}{{\cal Y}}
\newcommand{\cZ}{{\cal Z}}
\newcommand{\sQ}{{\mathsf{Q}}}
\newcommand{\sR}{{\mathsf{R}}}
\newcommand{\sD}{{\mathsf{D}}}
\newcommand{\ersp}{{\sf {ERSP}}}
\newcommand{\ANN}{{\sf {ANN}}}
\newcommand{\smp}{{\sf {SMP}}}
\newcommand{\pub}{{\sf {pub}}}
\newcommand{\tsigma}{\tilde{\sigma}}

\newcommand{\alice}{\sf {Alice}}
\newcommand{\bob}{\sf {Bob}}
\newcommand{\suppress}[1]{}

\newcommand{\mybox}[1]{
\begin{tabular}{|p{5.5in}|}
\hline 
\begin{center}
\begin{minipage}[t]{5.3in} 
#1 
\end{minipage}
\end{center}\\
\hline
\end{tabular}
}

\title{Optimal Direct Sum and Privacy Trade-off Results for Quantum and Classical
Communication Complexity}
\subtitle{[Full version]\footnote{A preliminary version of this paper
with many results mentioned here appeared in the 20th IEEE Conference
on Computational Complexity, 2005.}}
\author{
Rahul Jain\inst{1}\thanks{Research supported in part by ARO/NSA
USA. Part of this work was done while the author was at U.C. Berkeley,
Berkeley, USA, where it was supported by Army Research Office (ARO),
North California, under grant DAAD 19-03-1-00082. Part of the work
done while the author was at Tata Institute of Fundamental Research,
Mumbai, India.  }
\and
Jaikumar Radhakrishnan\inst{2}\thanks{Part of the work done while the
author was at Toyota Technological Institute Chicago, USA. } 
\and
Pranab Sen\inst{3}\thanks{Work done while the author was at University
of Waterloo, Waterloo, Canada.} 
\institute{
University of Waterloo, Waterloo, ON, Canada, N2L
3G1. \email{rjain@cs.uwaterloo.ca} \and 
Tata Institute of Fundamental Research, Mumbai, India.
\email{jaikumar@tifr.res.in} \and
Tata Institute of Fundamental Research, Mumbai, India.
\email{pgdsen@tcs.tifr.res.in}
} 
}
\date{}

\begin{document}
\maketitle
\thispagestyle{empty}
\begin{abstract}
We show optimal {\em Direct Sum} result for the {\em one-way
entanglement-assisted quantum communication complexity} for any
relation $f \subseteq \cX \times \cY \times \cZ$. We show:
$$ \sQ^{1,\pub}(f^{\oplus m}) = \Omega(m \cdot \sQ^{1,\pub}(f)),$$
where $\sQ^{1,\pub}(f)$, represents the one-way entanglement-assisted
quantum communication complexity of $f$ with error at most $1/3$ and
$f^{\oplus m}$ represents $m$-copies of $f$.
Similarly for the {\em one-way public-coin classical communication complexity} we show:
$$ \sR^{1,\pub}(f^{\oplus m}) = \Omega(m \cdot \sR^{1,\pub}(f)),$$
where $\sR^{1,\pub}(f)$, represents the one-way public-coin
classical communication complexity of $f$ with error at most $1/3$.
We show similar optimal Direct Sum results for the {\em Simultaneous
Message Passing} ($\smp$) quantum and classical models. For two-party two-way
protocols we present optimal {\em Privacy Trade-off} results leading
to a Weak Direct Sum result for such protocols. 

\vspace{0.2cm}
We show our Direct Sum and Privacy Trade-off results via {\em message
compression} arguments.  These arguments also imply a new {\em round
elimination} lemma in quantum communication, which allows us to extend
classical lower bounds on the {\em cell probe complexity} of some
{\em data structure problems}, e.g. {\em Approximate Nearest Neighbor
Searching} ($\ANN$) on the Hamming cube $\{0,1\}^n$ and {\em Predecessor
Search} to the quantum setting.

\vspace{0.2cm}
In a separate result we show that Newman's~\cite{Newman91}
technique of reducing the number of public-coins in a classical
protocol cannot be lifted to the quantum setting. We do this by
defining a general notion of {\em black-box} reduction of prior
entanglement that subsumes Newman's technique.  We prove that such a
black-box reduction is impossible for quantum protocols by exhibiting
a particular one-round quantum protocol for the {\em Equality} function
where the black-box technique fails to reduce the amount of prior
entanglement by more than a constant factor.

\vspace{0.2cm} In the final result in the theme of message compression,
we provide an upper bound on the problem of {\em Exact Remote State
Preparation} ($\ersp$).
\end{abstract}

\section{Introduction}
Communication complexity studies the communication required to
solve a computational problem in a distributed setting. Consider a
relation $f \subseteq \cX \times \cY \times \cZ$. In a two-party
protocol to solve $f$, one party say $\alice$ would be given input~$x \in
\cX$, and the other party say $\bob$ would be given input~$y \in
\cY$. The goal for $\alice$ and $\bob$ would be to communicate and find an
element~$z \in \cZ$ that satisfies the relation, i.e., to find a~$z$
such that~$(x,y,z) \in f$. The protocols they follow could be
deterministic, randomized or quantum leading to different notions of
deterministic, randomized and quantum communication complexity. Please
refer to Sec.~\ref{sec:prelimcomm} for detailed exposition to various
models, definitions and notations related to classical and quantum
communication complexity.

\subsection{Direct Sum} 
Let us consider a natural question in communication complexity as
follows. Suppose $\alice$ and $\bob$ wish to solve several, say $k$,
instances of relation $f$ simultaneously, with constant success on
the overall output. A {\em Direct Sum} theorem states that the
communication required for accomplishing this would be at least $k$
time the communication required for solving single instance of $f$,
with constant success. It is a natural and fundamental question in communication
complexity.

Although they seem highly plausible, it is well-known that Direct Sum
results fail to hold for some settings of communication. For example
for the {\em Equality} function ($\EQ_n$), in which $\alice$ and $\bob$ need
to determine if their $n$-bit inputs are equal or not, its randomized
private-coins communication complexity, denoted $\sR(\EQ_n)$ does not
satisfy the Direct Sum property.  It is known that $\sR(\EQ_n) =
\Theta(\log n)$ whereas for testing Equality of~$k = \log n$
\footnote{All logarithms in this article are taken to base 2 unless
otherwise specified.} pairs
of~$n$-bit strings $\sR(\EQ_n^{\oplus k}) = O(k\log k + \log n) = O(\log n \log\log n)$
(see, e.g., \cite[Example~4.3, page~43]{KushilevitzN97}), where we
might expect~$\sR(\EQ_n^{\oplus k})= \Omega(k\log n) = \Omega(\log^2
n)$. Similarly, Shaltiel~\cite{Shaltiel03} gives an example for which
a related notion called the {\em Strong Direct Product} property fails
to hold for average case (i.e., distributional) communication
complexity. (A Strong Direct Product theorem would show that even with
probability of success that is exponentially small in~$k$, the cost
of solving $k$ instances of $f$, would be~$k$ times the cost of solving one instance.)

\subsubsection{Previous works:} 
Notwithstanding these examples, Direct Sum results have met with some
success in several settings of communication. It is straightforward to
show that $\sD^1(f):$ the deterministic one-way communication complexity of every
relation $f$, satisfies the Direct Sum property. It is also known that for
two-way protocols, for any function $f$, $\sD(f^{\oplus k}) = \Omega(k
\cdot \sqrt{\sD(f)})$ (see, e.g.,~\cite[Exercise~4.11, page~46]{KushilevitzN97}).  For
classical distributional complexity, under the uniform distribution on
the inputs, Chakrabarti, Shi, Wirth, and Yao~\cite{chak:direct} showed
Direct Sum in the one-way and $\smp$ models of communication.  They
introduced an important notion of {\em information cost} and obtained
their Direct Sum result via a message compression argument. The
notion of information cost has also been effectively used to obtain
two-way classical and quantum communication complexity bounds for
example see~\cite{Bar-YossefJKS04,JainRS03b}.  Jain, Radhakrishnan,
and Sen~\cite{jain:icalp03} extended the result of~\cite{chak:direct},
and provided a Direct Sum result for classical distributional
complexity under any product distribution on inputs, for bounded-round
two-way protocols. They~\cite{jain:icalp03} again used the information
cost approach however achieved their message compression via different
techniques (than~\cite{chak:direct}) involving the {\em Substate
Theorem}~\cite{jain:substate}.  Recently, Harsha, Jain, McAllester,
and Radhakrishnan~\cite{HarshaJMR07} have strengthened the Direct Sum
result of~\cite{jain:icalp03} by reducing to a large extent its
dependence on the number of rounds. The message compression results in
~\cite{jain:icalp03,HarshaJMR07} have been used in the work of
Chakrabarti and Regev~\cite{ChakrabartiR04} to show lower bounds on
the Approximate Nearest Neighbor problem ($\ANN$) in the cell probe model.
P\v{a}tra\c{s}cu and Thorup~\cite{PatrascuT06} also use Direct Sum
type results to prove better lower bounds for this problem.

\subsubsection{Our results:} 
In this paper we prove that for any relation $f \subseteq \cX \times
\cY \times \cZ$, the classical
public-coin one-way communication complexity $\sR^{1,\pub}(f)$ and the
one-way entanglement assisted quantum communication complexity $\sQ^{1,\pub}(f)$
satisfy the Direct Sum property. Similarly in the $\smp$ model
$\sR^{||, \pub}(f)$ and $\sQ^{||,\pub}(f)$ satisfy the Direct Sum
property. Our precise results are as follows.

\begin{theorem}[Direct Sum]
Let $f \subseteq \cX \times \cY \times \cZ$ be a relation. Let
$\epsilon, \delta \in (0,1/2)$ with $\epsilon + \delta < 1/2$. 
For one-round protocols we have:
\begin{enumerate}
\item $\sR^{1, \pub}_\epsilon (f^{\oplus m}) 
\quad \geq \quad \Omega\left( \delta^3m \cdot \sR^{1,
\pub}_{\epsilon + \delta}(f)\right)$.
\item $\sQ^{1,  \pub}_\epsilon (f^{\oplus m}) 
\quad \geq \quad \Omega\left( \delta^3m \cdot \sQ^{1,\pub}_{\epsilon +
\delta}(f)\right)$. 
\end{enumerate}
Similarly for $\smp$ protocols (with shared resource as specified in
Section~\ref{sec:prelimcomm}), we have:
\begin{enumerate}
\item $\sR^{\|, \pub}_\epsilon (f^{\oplus m}) \quad \geq \quad 
 \Omega\left(\delta^3m \cdot \sR^{\|,\pub}_{\epsilon+ \delta} (f)\right).$
\item $\sQ^{\|, \pub}_\epsilon (f^{\oplus m}) \quad \geq \quad 
 \Omega\left(\delta^3m \cdot \sQ^{\|,\pub}_{\epsilon + \delta} (f)\right).$
\end{enumerate}
\end{theorem}

We obtain our Direct Sum results via message
compression results. Our message compression result for classical
one-way protocols is as follows:
 
\begin{result}[Classical one-way message compression, informal statement]
\label{res:classoneround}
Let $f\subseteq \cX \times \cY \times \cZ$ be a relation and let $\mu$
be a probability distribution (possibly non-product) on $\cX \times
\cY$.  Let $\cP$ be a one-way private-coins classical protocol for $f$
(with single message from $\alice$ to $\bob$) having bounded average error
under $\mu$.  Suppose $\alice$'s message in $\cP$ has {\em mutual
information} (please refer to Sec.~\ref{sec:prelim} for definition) at
most $k$ about her input. Then there is a
one-way deterministic protocol $\cP'$ for $f$ having similar average
error probability under $\mu$, in which $\alice$'s message is $O(k)$ bits
long.
\end{result}

We show similar message compression result for one-way quantum
protocols.
\begin{result}[Quantum one-way message compression, informal statement]
\label{res:quantoneround}
Let $f\subseteq \cX \times \cY \times \cZ$ be a relation and let $\mu$
be a probability distribution (possibly non-product) on $\cX \times
\cY$.  Let $\cP$ be a one-way quantum protocol without prior
entanglement for $f$ having bounded average error probability under
$\mu$.  Suppose $\alice$'s message in $\cP$ has mutual information at
most $k$ about her input. Then there is a one-way protocol $\cP'$ for
$f$ with prior entanglement having similar average error probability
under $\mu$, where $\alice$'s message is classical and $O(k)$ bits
long.
\end{result}

The proof of the above result  uses a technical quantum
information-theoretic fact called the Substate
Theorem~\cite{jain:substate}.  Essentially, it says that if a quantum
encoding of a classical random variable $x \mapsto \sigma_x$ has
information at most $k$ about $x$, then for most $x$,
$\frac{\sigma_x}{2^{O(k)}} \leq \sigma$ (for Hermitian matrices $A$,
$B$, $A \leq B$ is a shorthand for the statement ``$B - A$ is positive
semidefinite''), where $\sigma \defeq \E_x [\sigma_x]$. Similarly the
classical message compression result uses the classical version of the
Substate Theorem. The classical Substate Theorem was also used by
\cite{jain:icalp03} to prove their classical message compression
results.

Res.~\ref{res:quantoneround} also allows
us to prove a new {\em round elimination} result for quantum
communication. To state the round elimination
lemma, we first need the following definition.
\begin{definition}
\label{def:fk}
Let $f: \cX \times \cY \rightarrow \cZ$ be a function. 
The communication game $f^{(k),A}$ is defined as
follows: $\alice$ gets $k$ strings $x_1, \ldots, x_k \in \cX$.
$\bob$ gets an integer $j \in [k]$, a copy of strings 
$x_1, \ldots, x_{j - 1}$, and a string $y \in \cY$. They
are supposed to communicate and determine $f(x_j, y)$. The
communication
game $f^{(k),B}$ is defined analogously with roles of $\alice$ and
$\bob$ reversed.
\end{definition}
\begin{result}[Round elim., informal stmt.]
\label{res:roundelim}
Let $f: \cX \times \cY \rightarrow \cZ$ be a function. 
Suppose $\cP$ is a $t$-round quantum protocol for $f^{(k),A}$
with prior entanglement having bounded worst case error.
Suppose $\alice$ starts the communication and the first and
second  messages of $\cP$ are $l_1$ and $l_2$ qubits 
long respectively.
Then there is a $(t - 1)$-round protocol for $f$ with prior
entanglement having similar worst case error where $\bob$ starts
the communication and the first message is
$l_2 \cdot 2^{O(l_1/k)}$ qubits long. The subsequent communication
in $\cP'$ is similar to that in $\cP$.
\end{result} 
The classical analogue of the above result was shown by Chakrabarti
and Regev~\cite{ChakrabartiR04}, where they used the message compression
arguments of~\cite{jain:icalp03,HarshaJMR07} to arrive at their result.
The above round elimination lemma is useful in situations where
$\alice$'s message length $l_1$ is much smaller than $\bob$'s message
length $l_2$. Such a situation arises in proving cell probe lower
bounds for data structure problems like $\ANN$ in $\{0, 1\}^n$ and Set
Predecessor. \cite{ChakrabartiR04} used it crucially in proving optimal
randomized cell probe lower bounds for $\ANN$. Recently, Patrascu and
Thorup~\cite{patrascu:pred,PatrascuT06} used the same classical technique to prove
sharper lower bounds for the Set Predecessor problem. We remark that
both these results carry over to the {\em address-only quantum cell probe
model} (defined in \cite{sen:icalp01}) as a consequence of Res.~\ref{res:roundelim}.

\subsection{Privacy trade-offs}
Let us consider another natural question in communication complexity
as follows. Let $f\subseteq \cX \times \cY \times \cZ$ be a relation.
We are interested in the {\em privacy loss} of $\alice$ and $\bob$ that is
inherent in computing $f$. Privacy in communication complexity was
studied in the classical setting by Bar-Yehuda et
al.~\cite{baryehuda:privacy}, and in the quantum setting by
Klauck~\cite{klauck:privacy} and Jain, Radhakrishnan, and
Sen~\cite{jain:substate}.  For studying privacy issues in quantum
communication, we only consider protocols without prior entanglement.
To define the privacy loss of $\alice$, imagine that $\alice$ follows the
protocol $\cP$ honestly but $\bob$ is malicious and deviates arbitrarily
from $\cP$ in order to extract the maximum amount of information about
$\alice$'s input. The only constraint on $\bob$ is that $\alice$ should not be
able to figure out at any point of time whether he is cheating or not;
we call such a cheating strategy of $\bob$ {\em undetectable}.  Suppose
$\mu = \mu_\cX \times \mu_\cY$ is a product probability distribution
on $\cX \times \cY$.  Let register $X$ denote the input qubits of
$\alice$, and $B$ denote all the qubits in the possession of $\bob$ at the
end of $\cP$.  We assume the input registers of $\alice$ and $\bob$ are
never modified and are never sent as messages in $\cP$.  Then the
privacy loss of $\alice$ under distribution $\mu$ at the end of $\cP$ is
the maximum mutual information $I(X : B)$ over all undetectable
cheating strategies of $\bob$. The privacy loss of $\bob$ can be defined
analogously.  In the quantum setting $\bob$ has a big bag of undetectable
cheating tricks that he can use in order to extract information about
$X$. For instance, he can start the protocol $\cP$ by placing a
superposition of states $\ket{\mu_\cY}$ (for a probability
distribution $\pi$ on $\cZ$, $\ket{\pi} \defeq \sum_z \sqrt{\pi(z)}
\ket{z}$) in his input register $Y$ and running the rest of the
protocol honestly. This trick works especially well for so-called
{\em clean} protocols that leave the work qubits of $\alice$ and $\bob$ at the
end of the protocol in the state $\ket{0}$. For example, consider the
following exact clean protocol $\cP$ computing the inner product
modulo $2$, $x \cdot y$, of two bit strings $x, y \in \{0, 1\}^n$:
$\alice$ sends her input $x$ to $\bob$, $\bob$ computes $x \cdot y$ and sends
back $x$ to $\alice$ keeping the bit $x \cdot y$ with himself, and
finally $\alice$ zeroes out $\bob$'s message by XORing with her input $x$.
If $\bob$ does the above `superposition cheating' trick for $\cP$, his
final state at the end of $\cP$ becomes $\left(\sum_{y \in \{0, 1\}^n}
\ket{y, x \cdot y}\right)$.  It is easy to see that $\bob$ has $\frac{n}{2}$ bits
of information about $x$, if $x$ is distributed uniformly in $\{0,
1\}^n$. Thus, the privacy loss from $\alice$ to $\bob$ for this protocol is
at least $\frac{n}{2}$, under the uniform distribution on $\{0, 1\}^n
\times \{0, 1\}^n$.  See \cite{cleve:ip} for more details. Thus, it is
conceivable that $\alice$ and $\bob$ use an `unclean' protocol to compute
$f$ in order to minimize their privacy losses. We shall be concerned
with proving tradeoffs between the privacy losses of $\alice$ and $\bob$ for
any quantum protocol computing $f$, including `unclean' ones. Please
refer to Sec.~\ref{sec:multiround}, Def.~\ref{def:privacyquant} for precise
definition of privacy loss. Note
that defining the privacy loss only for quantum protocols without
prior entanglement is without loss of generality, since we can convert
a protocol with prior entanglement into one without prior entanglement
by sending the entanglement as part of the first message of the
protocol; this process does not affect the privacy loss since after
the first message is sent, the qubits in the possession of $\alice$ and
$\bob$ are exactly the same as before.

For private-coin randomized classical protocols, a related notion 
called {\em information cost}, was defined 
in \cite{chak:direct,Bar-YossefJKS04} to be the mutual information
$I(X Y : M)$ between the players' inputs and the complete message
transcript $M$ of the protocol. For quantum protocols there is no
clear notion of a message transcript, hence we use our definition
of privacy instead. Also, other than cryptographic reasons 
there is also another reason why we allow 
the players to use undetectable cheating strategies. In the
above clean protocol $\cP$ for the inner product function, if
both $\alice$ and $\bob$ were honest the final state of $\cP$ would
be $\ket{x} \otimes \ket{y, x \cdot y}$, where the first state
belongs to $\alice$ and the second to $\bob$. Under the uniform distribution
on $x, y$ the privacy loss from $\alice$ to $\bob$
is $1$, whereas the classical information cost is at least $n$.
This shows that in the quantum setting, because of the ability of
players to `forget' information by uncomputing, it is better to
allow undetectable cheating strategies for players in the definition
of privacy loss in order to bypass examples such as the above. 

\subsubsection{Our results:} In this paper we relate the privacy loss
incurred in computation of any relation $f$ to the one-way
communication complexity $f$. We show that in multi-round protocols
with low privacy loss, all the messages could be replaced by a single
short message.  For quantum protocols, again using the Substate
Theorem~\cite{jain:substate}, we prove the following result.
\begin{result}[Quantum multiple rounds compression, informal stmt.]
\label{res:quantmultiround}
Let $f \subseteq \cX \times \cY \times \cZ$ be a relation and let
$\mu$ be a product probability distribution on $\cX \times \cY$.  Let
$\cP$ be a multi-round two-way quantum protocol without prior
entanglement for $f$ having bounded average error probability under
$\mu$.  Let $k_a$, $k_b$ denote the privacy losses of $\alice$ and $\bob$
respectively under distribution $\mu$ in $\cP$. Then there is a
one-way protocol $\cP'$ for $f$ with prior entanglement having similar
average error probability under $\mu$, such that the single message of
$\cP'$ is from $\alice$ to $\bob$, it is classical and $k_a 2^{O(k_b)}$ bits
long. Similarly statement also holds with the roles of $\alice$ and
$\bob$ reversed. 
\end{result}
We would like to remark that Res.~\ref{res:quantoneround} does
not follow from Res.~\ref{res:quantmultiround}.
Res.~\ref{res:quantoneround} holds for any probability distribution
on $\cX \times \cY$ whereas our proof of Res.~\ref{res:quantmultiround}
requires product distributions.  It is open whether a similar multi-round
compression result can be proved for non-product distributions for
quantum protocols. 

Similarly for classical protocols we show the following result. Please
refer to Sec.~\ref{sec:multiround}, Def.~\ref{def:privacyclass} for
precise definition of privacy loss for classical protocols. 
\begin{result}[Classical multiple rounds compression, informal stmt.]
\label{res:classmultiround}
Let $f \subseteq \cX \times \cY \times \cZ$ be a relation and let
$\mu$ be a product probability distribution on $\cX \times \cY$.  Let
$\cP$ be a multi-round two-way private-coins classical protocol for
$f$ having bounded average error probability under $\mu$.  Let $k_a$,
$k_b$ denote the privacy losses of $\alice$ and $\bob$ respectively under
distribution $\mu$ in $\cP$. Then there is a one-way deterministic
protocol $\cP'$ for $f$ having similar average error probability under
$\mu$, such that the single message of $\cP'$ is from $\alice$ to $\bob$ and
is $k_a 2^{O(k_b)}$ bits long. Similarly statement also holds with the
roles of $\alice$ and $\bob$ reversed.
\end{result}
We would like to point out that the proof of this result does not follow entirely
on the lines of Res.~\ref{res:quantmultiround}, essentially due to
the difference in the definition between the notions of privacy loss
for classical and quantum protocols. Therefore its proof is presented
separately.

These message compression results immediately imply the following
privacy trade-off results (similar results hold with the roles of
$\alice$ and $\bob$ reversed.) 
\begin{result}[Quantum privacy trade-off]
\label{res:privacyquant}
Let the privacy loss of $\alice$ be $k_a$ and the privacy loss of $\bob$ be $k_b$ 
at the end of a quantum protocol without entanglement $\cP$ for computing a relation
$f\subseteq \cX \times \cY \times \cZ$. Then,$$k_a 2^{O(k_b)} \quad
\geq \quad \sQ^{1, A \rightarrow B, {\pub}, [\;]} (f),$$ where $\sQ^{1, A
\rightarrow B, {\pub}, [\;]} (f)$ is the maximum over all product
distributions $\mu$ on $\cX \times \cY$, of the one-round quantum
communication complexity (with $\alice$ communicating) of $f$ with prior
entanglement having bounded average error under $\mu$.
\end{result}

\begin{result}[Classical privacy trade-off]
Let the privacy loss of $\alice$ be $k_a$ and the privacy loss of $\bob$ be $k_b$ 
at the end of a classical private coins protocol $\cP$ for computing a relation
$f\subseteq \cX \times \cY \times \cZ$. Then,$$k_a 2^{O(k_b)} \quad
\geq \quad \sR^{1, A \rightarrow B, [\;]} (f),$$ where $\sR^{1, A
\rightarrow B, [\;]} (f)$ is the maximum over all product
distributions $\mu$ on $\cX \times \cY$, of the one-round classical
distributional communication complexity (with $\alice$ communicating) of
$f$ having bounded average error under $\mu$. 
\end{result}

\paragraph{Remarks:}
\begin{enumerate}
\item Note that Res.~\ref{res:privacyquant} also shows that the privacy loss for
computing $f$ is lower bounded by $\Omega(\log \sQ^{1, {\pub}, [\;]}
(f))$. This latter result can be viewed as the privacy analogue of
Kremer's result~\cite{kremer:quantcc} that the bounded error quantum
communication complexity of $f$ is lower bounded by the logarithm of
its deterministic one-round communication complexity.
\item Res.~\ref{res:quantmultiround} and
Res.~\ref{res:classmultiround} also allow us to show  weak general 
Direct Sum result for quantum protocols and classical protocols as mentioned in
Corr.~\ref{corr:weakdirectquant} and Corr.~\ref{corr:weakdirectclass}
respectively in Sec.~\ref{sec:multiround}.
\item All these results are optimal in general as evidenced by the Index function
problem~\cite{nayak:index}. In the Index function problem,
$\alice$ is given a database $x \in \{0, 1\}^n$ and $\bob$ is given an index
$i \in [n]$. They have to communicate and determine $x_i$, the $i$-th
bit of $x$. The
one-round quantum communication complexity from $\alice$ to $\bob$ for this
problem is $\Omega(n)$, even for bounded average error under the
uniform distribution and in the presence of prior entanglement.  Thus,
we get the privacy tradeoff $k_a 2^{O(k_b)} = \Omega(n)$ for the Index
function problem. This is optimal; consider a deterministic protocol
where $\bob$ sends the first $b$ bits of his index and $\alice$ replies by
sending all the $\frac{n}{2^b}$ bits of her database consistent with
$\bob$'s message.  
\item Earlier, Jain, Radhakrishnan, and
Sen~\cite{jain:substate} had proved the same privacy tradeoff for the
Index function problem specifically. Our general tradeoff above can be
viewed as an extension of their result to all functions and relations.
\end{enumerate}

\subsection{Impossibility of black-box entanglement reduction}
Let us return to the third main question we investigate in this work
which appears different but is intimately related to the theme of
message compression and we mention this connection later.

We know that for some quantum communication problems, presence of
prior entanglement helps in reducing the communication. For example,
the technique of superdense coding~\cite{bennett:superdense} allows us
to often reduce the communication complexity by a multiplicative
factor of $2$.  So a natural question that arises is how much prior
entanglement is really required by a quantum protocol? For
classical communication, Newman~\cite{Newman91} has shown that
$O(\log n)$ shared random bits are sufficient for any protocol. This
is tight, as evidenced by the  Equality function on $\{0, 1\}^n$
which requires $\Theta(\log n)$ bits with private randomness and
$O(1)$ bits with shared randomness. One might hope to extend
Newman's~\cite{Newman91} proof that a classical protocol needs
only $O(\log n)$ shared random bits to the quantum setting.  Newman's
proof uses a Chernoff-based sampling argument on the shared random bit
strings to reduce their number to $O(n)$.  Moreover, the reduction is done in a
black-box fashion i.e. it does not change the computation of
$\alice$ and $\bob$ in the protocol.  In the quantum setting, one might
similarly hope to reduce the amount of entanglement of the prior
entangled state $\ket{\phi}$ to $O(\log n)$ and leave the unitary
transforms of $\alice$ and $\bob$ unaffected i.e. the hope is to find a
black-box Newman-style prior entanglement reduction technique. We show
that such a black-box reduction is impossible.

To state our result precisely, we need the following definitions.
\begin{definition}[Similar protocols]
Two protocols $\cP$ and $\cP'$ with prior
entanglement and outputting values in $\cZ$ are called {\em similar
protocols} if both use the same number of qubits and the same
unitary transformations and
measurements, have the same amount of communication and 
for all $(x,y) \in \{0, 1\}^n \times \{0, 1\}^n$, 
$\totvar{\cP(x, y) - \cP'(x, y)} < 1/20$.
Here, $\cP(x, y)$, $\cP'(x, y)$ are the probability distributions
on $\cZ$ of the output of protocol $\cP$, $\cP'$ on input $(x, y)$.
$\cP$, $\cP'$
may use different quantum states as their input independent prior
entanglement.
\end{definition}
\begin{definition}[Amt. of entanglement]
\label{def:ent}
For a bipartite pure state $\ket{\phi}_{A B}$, consider its 
{\em Schmidt decomposition},
$\ket{\phi} = \sum_{i=1}^{k} \sqrt{\lambda_i} \ket{a_i} \otimes
              \ket{b_i}$,
where $\{a_i\}$ is an orthonormal set and so is $\{b_i\}$,
$\lambda_i \geq 0$ and $\sum_i \lambda_i = 1$.
The {\em amount of entanglement} of $\ket{\phi}_{A B}$ is defined
to be 
$E(\ket{\phi}_{A B}) \defeq -\sum_i \lambda_i \log \lambda_i$.
The {\em Schmidt rank} of $\ket{\phi}_{A B}$ is defined to be $k$.
\end{definition}
One might hope that the following conjecture is true.
\begin{conjecture}
For any protocol $\cP$ for 
$f: \{0, 1\}^n \times \{0, 1\}^n \rightarrow \cZ$ with prior
entanglement, there exists a similar protocol $\cP'$
that starts with prior entanglement
$\ket{\phi}_{A B}$, $E(\ket{\phi}_{A B}) = O(\log n)$.
\end{conjecture}
We prove that the above conjecture is {\bf not} correct for quantum
communication protocols. 
\begin{result}[No black-box red. of prior entang.]
\label{res:nobboxred}
Let us denote the Equality function on $n$-bit strings by $\EQ_n$ . 
There exists a one-round quantum protocol $\cP$ for $\EQ_n$ with 
$\frac{2n}{3} + \log n + \Theta(1)$ EPR pairs of prior 
entanglement and communicating $4$ bits,
such that there is no
similar protocol $\cP'$ that starts with a prior entangled
state $\ket{\phi}_{A B}$, $E(\ket{\phi}_{A B}) \leq n/600$.
\end{result}

Our proof of this result follows essentially by sharpening the
geometric arguments behind the proof of the `recipient-non-invasive
incompressibility' result of Jain, Radhakrishnan, and
Sen\cite{jain:icalp03}. Jain, Radhakrishnan, and
Sen~\cite{jain:icalp03} showed that for classical constant round
private-coin protocols with a product probability distribution on
their inputs, one can compress the messages to the information cost of
the protocol.  Their compression technique for classical protocols was
`recipient-non-invasive' in the sense that, for one round protocols,
it did not change the computation of the recipient except up to a
trivial relabeling of the messages.  They however also showed
that such a recipient-non-invasive compression result does not hold
for quantum protocols; they exhibited a one-round quantum protocol
without prior entanglement for the Equality function on $n$-bit
strings with constant privacy loss, where any recipient-non-invasive
compression strategy cannot compress $\alice$'s message by more than a
multiplicative factor of $6$! We essentially convert their
``incompressibility of message'' result to ``incompressibility of
prior-entanglement'' result.

\paragraph{Remarks:}
\begin{enumerate}
\item The above Res.~\ref{res:nobboxred} shows that in order to reduce prior
entanglement in quantum communication, one has to look beyond
black-box arguments and change the unitary transforms of $\alice$
and $\bob$. 

\item Recently Gavinsky~\cite{Gavinsky06} showed that even if $\alice$
and $\bob$ are allowed 
to change their operations, there is an exponential increase that can
occur in the required message length for computation of a relation in
case the prior-entanglement is reduced only by a logarithmic
factor. However Gavinsky measures shared entanglement with the number of
qubits in the shared state between $\alice$ and $\bob$, and not with
the measure of entanglement as considered by us in
Def.~\ref{def:ent}. Hence Res.~\ref{res:nobboxred}, which 
first appeared in~\cite{JainRS05}, is incomparable to Gavinsky's result.

\end{enumerate}

\subsection{Exact Remote State Preparation ($\ersp$)} 
The final result we present in the theme of message compression
concerns the communication complexity of the Exact Remote State
Preparation ($\ersp$) problem. The $\ersp$ problem is as follows. 
Let $E: x \rightarrow \rho_x$ be an encoding from a set $\cX$ to the
set of quantum states. \\
{\bf Problem $\ersp(E)$:}
\begin{enumerate}
\item $\alice$ and $\bob$
start with prior entanglement. 
\item $\alice$ gets $x \in \cX$.
\item They interact at the end of which $\bob$ should end up with
$\rho_x$ in some register.
\end{enumerate}
  
\paragraph{Remark:} The adjective 'exact' signifies that we do not
allow for any fidelity loss in the state that $\bob$ should end up. 

We provide the following upper bound on the communication complexity
of this problem.
\begin{theorem}
\label{thm:ersp}
Let $E: x \rightarrow \rho_x$ be an encoding where $\rho_x$ is a pure
state for all $x$ and let $\sigma$ be any state with full rank. There is a protocol
$\cP$ for $\ersp(E)$ with expected communication bounded by $\max_x
\{\log (\Tr \sigma^{-1} \rho_x) + 2\log \log (\Tr \sigma^{-1} \rho_x)\}$.
\end{theorem}

\subsection{Organization of the paper}
In the next section, we collect some preliminaries that will 
be required in the proofs of the message compression results. 
In Sec.~\ref{sec:oneround}, we prove our results on first
round compression and round elimination in quantum protocols.
We prove our multi-round compression result
in Sec.~\ref{sec:multiround}.
In Sec.~\ref{sec:nobboxred},
we show that black-box reduction of prior entanglement
in quantum communication is impossible. Finally in Sec.~\ref{sec:rsp}
we provide the proof of the upper bound on the $\ersp$ problem.

\section{Preliminaries}
\label{sec:prelim}
\subsection{Information Theory} A quantum state is a positive semi
definite trace one operator. For a quantum state $\rho$, its {\em von-Neumann
 entropy} is defined as $S(\rho) \defeq \sum_i -\lambda_i \log
\lambda_i$, where $\lambda_i$s represent the various eigenvalues of
$\rho$.  For an $l$ qubit quantum system $A$, $S(A) \leq l$.  For
correlated quantum systems $A,B$ their mutual information is defined
as $I(A:B) = S(A) + S(B) - S(AB)$. Given a tri-partite system $A,B,C$,
mutual information satisfies the {\em monotonicity property} that is
$I(A:BC) \geq I(A:B)$.  Let us define $I(A:B | C) \defeq I(A:BC) -
I(A:C)$. We have the following very useful {\em Chain
Rule} for mutual information.  $$I(A:B_1 \ldots B_k) = \sum_{i=1}^k
I(A:B_i | B_1\ldots B_{i-1}).$$ Therefore if $B_1$ through $B_k$ are
independent systems then, $$ I(A:B_1 \ldots B_k) \geq \sum_{i=1}^k I(A:B_i).$$ 

For classical random variables the analogous definitions and facts hold
{\em mutates mutandis} and we skip making explicit statements here for brevity.
\subsection{Communication complexity} 
\label{sec:prelimcomm}
\subsubsection{Quantum communication complexity:}  Consider a 
two-party quantum communication protocol $\cP$ for computing a
relation $f : \{0, 1\}^n \times \{0, 1\}^n
\rightarrow \cZ$. The relations we consider are always total in the sense that for
every $(x,y) \in \cX \times \cY$, there is at least one $z \in
\cZ$, such that $(x,y,z) \in f$. In a two-way protocol $\cP$ for
computing $f$, $\alice$ and $\bob$ get inputs $x \in \cX$ and $y \in
\cY$ respectively. They send messages (qubits) to each other in turns, and
their intention is to determine an answer $z \in \cZ$ such that
$(x,y,z) \in f$. We assume that $\cP$ starts with the internal
work qubits of $\alice$ and $\bob$ in the state $\ket{0}$. Both the parties
use only unitary transformations for their internal computation,
except at the very end when the final recipient of a message makes a
von-Neumann measurement of some of her qubits to declare the
output. Thus, the joint state of $\alice$ and $\bob$ is always pure during
the execution of $\cP$.  We also assume that the players make {\em
safe} copies of their respective inputs using CNOT gates before
commencing the protocol. These safe copies of the inputs are not
affected by the subsequent operations of $\cP$, and are never sent as
messages.  In this paper, we consider protocols with and without prior
entanglement.  By prior entanglement, we mean a pure quantum state
$\ket{\phi}$ that is shared between $\alice$ and $\bob$ and that is
independent of their input $(x, y)$.  The state $\ket{\phi}$ can be supported on
an extremely large number of qubits. The unitary transforms of $\alice$
in $\cP$ are allowed to address her share of the qubits of
$\ket{\phi}$; similarly for $\bob$.  The classical analogue of prior
entanglement is shared random bits.  Often, the prior entanglement in
a quantum protocol is in the form of some number of EPR pairs,
one-half of which belongs to $\alice$ and the other half belongs to $\bob$.

Given $\epsilon \in (0,1/2)$, the two-way quantum communication
complexity $\sQ_\epsilon(f)$ is defined to be the communication of the
best two-way quantum protocol without prior entanglement, with error
at most $\epsilon$ on all inputs.  Whenever error parameter $\epsilon$
is not specified it is assumed to be $1/3$. Given a distribution $\mu$ on
$\cX \times
\cY$, we can similarly define the quantum  distributional two-way
communication complexity of $f$, denoted $\sQ^{\mu}_\epsilon(f)$,  to
be the communication of the best one-way quantum protocol without
entanglement for $f$, such that the average error of the protocol over
the inputs drawn from the distribution $\mu$ is at most
$\epsilon$. We
define~$\sQ^{[\;]}_{\epsilon}(f) \defeq \max_{\mu \textrm{ product}}
\sQ_{\epsilon}^{\mu}(f)$. The corresponding quantities for protocols
with entanglement are denoted with the superscript $\pub$.

The following result due to Yao~\cite{Yao77} is a very useful fact
connecting worst-case and distributional communication
complexities. It is a consequence of the {\em MiniMax\/} theorem in
game theory~\cite[Thm.~3.20, page~36]{KushilevitzN97}.
\begin{lemma}[Yao's principle~\cite{Yao77}]
\label{lem:yao} $\sQ^{\pub}_{\epsilon}(f) = \max_{\mu}
\sQ_{\epsilon}^{\pub, \mu}(f)$.
\end{lemma}

Similar relationships as above also hold in the various other models that we
mention below {\em mutates mutandis}.

In the one-way protocols we consider, the single message is always
assumed to be from $\alice$ to $\bob$ unless otherwise
specified. Sometimes we specify the direction of the message for example with
superscript $A\rightarrow B$. The
complexities $\sQ^1_\epsilon(f), \sQ^{1, \pub}_\epsilon(f),
\sQ^{1,\mu}_\epsilon(f), \sQ^{1,[\;]}_{\epsilon}(f) $ could be
analogously defined in the one-way case.
 
In the Simultaneous Message Passing ($\smp$) model, $\alice$ and
$\bob$ each send a message each to a third party called {\sf
Referee}. In the $\smp$ protocols we consider, we let prior
entanglement to be shared between $\alice$ and {\sf Referee} and
$\bob$ and {\sf Referee} and public coins to be shared between
$\alice$ and $\bob$. The communication complexity in this described
model is denoted by $\sQ^{\|, \pub}_\epsilon (f)$.

\subsubsection{Classical communication complexity:} Let us now
consider classical communication protocols.  We let $\sD(f)$
represent the deterministic two-way communication complexity, that is
the communication of the best deterministic two-way protocol computing $f$
correctly on all inputs.  Let~$\mu$ be a probability distribution
on~$\cX \times \cY$ and $\epsilon \in
(0,1/2)$.  We let $\sD_{\epsilon}^{\mu}(f)$ represent the
distributional two-way communication complexity of $f$ under $\mu$
with expected error $\epsilon$, i.e., the communication of the best
private-coin two-way protocol for $f$, with distributional error
(average error over the coins and the inputs) at most $\epsilon$ under
$\mu$.  It is easily noted that $\sD_{\epsilon}^{\mu}(f)$ is always
achieved by a deterministic two-way protocol, and henceforth we will 
restrict ourselves to deterministic protocols in the context of
distributional communication complexity.  We let
$\sR^{\pub}_{\epsilon}(f)$ represent the public-coin randomized
two-way communication complexity of $f$ with worst case error
$\epsilon$, i.e., the communication of the best public-coin randomized
two-way protocol for $f$ with error for each input $(x,y)$ being at
most~$\epsilon$.  The analogous quantity for private coin randomized
protocols is denoted by $\sR_{\epsilon}(f)$. The public- and
private-coin randomized communication complexities are not much
different, as shown in Newman's result~\cite{Newman91} that
\begin{equation}
	\sR(f) = O(\sR^{\pub}(f)+ \log\log |\cX| + \log \log |\cY|).
\end{equation} 

We define~$\sR^{[\;]}_{\epsilon}(f) \defeq \max_{\mu \textrm{ product}}
\sD_{\epsilon}^{\mu}(f)$. The analogous communication complexities for
one-way protocols could be similarly defined. As before, we put superscript 1 to
signify that they stand for one-way protocols and superscript $\|$ to
signify $\smp$ protocols. In classical public
coin $\smp$ protocols that we consider, we let the public coins to be
shared between $\alice$  and $\bob$.


\subsection{Substate Theorem and $(\delta,\alpha)$-corrector}
All our message compression arguments are based on the following
common idea: If $\alice$ does not reveal much information about her
input, then it must be the case that $\bob$'s state after receiving
$\alice$'s messages does not vary much (as $\alice$'s input varies). In this
situation, $\alice$ and $\bob$ can start in a suitable input independent
state and $\alice$ can account for the variation by applying appropriate
local transformations on her registers. We formalize this idea using
the notion of a $(\delta, \alpha)$-corrector, and establish the
existence of such correctors by appealing to the following information-theoretic
result, the Substate Theorem due to Jain, Radhakrishnan,
and Sen~\cite{jain:substate}.
\begin{fact}[Substate Theorem, \cite{jain:substate}]
\label{fact:substate}
Let $\cH, \cK$ be two finite dimensional Hilbert spaces and 
$\dim(\cK) \geq \dim(\cH)$. Let $\C^2$ denote the two dimensional
complex Hilbert space.  Let $\rho, \sigma$ be density matrices in
$\cH$ such that $S(\rho \| \sigma) < \infty$.  Let 
$\ket{\overline{\rho}}$ be a
purification of $\rho$ in $\cH \otimes \cK$. Then, for $r > 1$, there
exist pure states $\ket{\phi}, \ket{\theta} \in \cH
\otimes \cK$ and 
$\ket{\overline{\sigma}} \in \cH \otimes \cK \otimes \C^2$, 
depending on $r$, such
that $\ket{\overline{\sigma}}$ is a purification of $\sigma$ and
$\trnorm{\ketbra{\overline{\rho}} - \ketbra{\phi}} \leq 
 \frac{2}{\sqrt{r}}$, where
\begin{displaymath}
\ket{\overline{\sigma}} \defeq
\sqrt{\frac{r-1}{r 2^{r c}}} \, \ket{\phi}\ket{1} +
\sqrt{1 - \frac{r-1}{r 2^{r c}}} \, \ket{\theta}\ket{0}
\end{displaymath}
and $c \defeq S(\rho \| \sigma) + O(\sqrt{S(\rho \| \sigma)}) + O(1)$. Note
that one can, by means of a local unitary operator on $\cK \otimes
\C^2$, transform any known purification $\ket{\overline{\sigma'}}$ of
$\sigma$ to $\ket{\overline{\sigma}}$.  Also, measuring the last qubit
of $\ket{\overline{\sigma}}$ and observing a $\ket{1}$ puts the
remaining qubits into the state $\ket{\phi}$. It follows that for
every purification $\ket{\overline{\sigma'}}$ of $\sigma$, there is an
unnormalized superoperator $\cM$, depending on
$\ket{\overline{\sigma'}}$, acting on the qubits of
$\ket{\overline{\sigma'}}$ other than those of $\sigma$, such that
$\cM(\ketbra{\sigma'})$ normalized is equal to $\ket{\phi}$.
Furthermore, this superoperator succeeds with probability at least
$\frac{r - 1}{r 2^{r c}}$.
\end{fact}
\begin{definition}[$(\delta, \alpha)$-corrector]
\label{def:corrector}
Let $\alice$ and $\bob$ form a bipartite quantum system.  Let $X$ denote
$\alice$'s input register, whose values range over the set $\cX$.  For
$x\in \cX$, let $\sigma_x$ be a state wherein the state of the 
register
$X$ is $\ket{x}$; that is, $\sigma_x$ has the form $\ketbra{x} \otimes
\rho_x$. Let $\mu$ be a probability distribution on $\cX$.  Let
$\sigma$ be some other joint state of $\alice$ and $\bob$. A $(\delta,
\alpha)$-corrector for the ensemble 
$\{\{\sigma_x\}_{x \in \cX}; \sigma\}$ with
respect to the distribution $\mu$ is a family of unnormalized
superoperators $\{\cM_x\}_{x \in \cX}$ acting only on $\alice$'s 
qubits such that:
\begin{enumerate}
\item $r_x \defeq \Tr \cM_x(\sigma)= \alpha$ for all 
      $x \in \cX$, that is,
      $\cM_x$ when applied to $\sigma$ succeeds with probability 
      exactly $\alpha$.

\item $\cM_x(\sigma)$ has the form 
      $\ketbra{x} \otimes \rho'_x$, that is,
      the state of the register $X$ of $\alice$ is $\ket{x}$ when
      $\cM_x$ succeeds.

\item $\E_\mu\left[\trnorm{\sigma_x - \frac{1}{\alpha}\cM_x(\sigma)}
             \right] \leq \delta$, that is, $\cM_x$ on success
      {\em corrects} 
      the state $\sigma$ by bringing it to within trace distance
      $\delta$ from $\sigma_x$.
\end{enumerate}
\end{definition}

We shall also need the following observation.
\begin{proposition}
\label{prop:trdistp}
Suppose a boolean-valued measurement $\cM$ succeeds with 
probabilities
$p$, $q$ on quantum states $\rho$, $\sigma$ respectively. Let 
$\rho'$, $\sigma'$ be the respective quantum states if $\cM$ 
succeeds. Then, 
$\trnorm{\rho' - \sigma'} \leq \frac{1}{\max\{p,q\}} 
                               \trnorm{\rho - \sigma}$.
\end{proposition}
\begin{proof}
We formalize the intuition that if some measurement distinguishes 
$\rho'$ and $\sigma'$, then there is a measurement that 
distinguishes $\rho$ and $\sigma$.
Assume $p \geq q$ (otherwise interchange the roles of $\rho$ and
$\sigma$). Now there exists (see e.g.~\cite{aharonov:mixed}) 
an orthogonal projection $M'$, such that 
$\Tr M'(\rho' - \sigma') = \frac{\trnorm{\rho' - \sigma'}}{2}$.
Let $M''$ be the POVM element obtained by first applying POVM
$\cM$ and
on success applying $M'$. Then the probability of success of $M''$ on
$\rho$ is $p \cdot \Tr M'\rho'$, and the probability of success of 
$M''$ on $\sigma$ is 
$q \cdot \Tr M'\sigma' \leq p \cdot \Tr M'\sigma'$.
Thus, 
\begin{eqnarray*}
\frac{1}{2}\trnorm{\rho - \sigma} 
& \geq & \Tr M''\rho - \Tr M''\sigma \\
& \geq & p (\Tr M'\rho' - \Tr M'\sigma') \\
&   =  & \frac{p}{2} \cdot \trnorm{\rho' - \sigma'},
\end{eqnarray*}
implying that 
$\trnorm{\rho' - \sigma'} \leq \frac{\trnorm{\rho - \sigma}}{p}$.
\qed
\end{proof}

We are now ready to use the Substate Theorem to
show the existence of good correctors when $\bob$'s
state does not contain much information about $\alice$'s input.
While applying the Substate Theorem below,
it will be helpful to think of $\alice$'s Hilbert space as
$\cK \otimes \C^2$ and $\bob$'s Hilbert space as $\cH$ in
Fact~\ref{fact:substate}.
\begin{lemma}
\label{lem:corrector}
For $x \in \cX$, let $\ket{\phi_x}\defeq\ket{x}\ket{\psi_x}$ be a
joint pure state of $\alice$ and $\bob$, where $\ket{x}$ 
and possibly some other qubits of
$\ket{\psi_x}$ belong to $\alice$'s subsystem $A$, and the remaining
qubits of $\ket{\psi_x}$ belong to $\bob$'s subsystem $B$. Let $\mu$ 
be a probability
distribution on $\cX$; let $\sigma \defeq \E_\mu \ketbra{\phi_x}$ and
$\ket{\phi} \defeq \sum_x \sqrt{\mu(x)} \ket{\phi_x}$.
Let $X$ denote the register of $\alice$ containing $\ket{x}$. 
Suppose $I(X : B)=k$, when the joint state 
of $A B$ is $\sigma$.  Then for
$\delta > 0$, there is a $(\delta, \alpha)$-corrector
$\{\cM_x\}_{x \in \cX}$ for the ensemble
$\{\{\ket{\phi_x}\}; \ket{\phi}\}$ where 
$\alpha = 2^{-O(k/\delta^3)}$.
\end{lemma}
\begin{proof}
Let $\rho_x \defeq \Tr_A \ketbra{\phi_x}$ and $\rho \defeq \Tr_A
\ketbra{\phi}$. Note that $\rho = \E_{\mu} \rho_x$. Now, 
$k = I(X : B) = \E_{\mu} S(\rho_x \| \rho)$. 
By Markov's inequality, there is a
subset $\Good \subseteq \cX$, $\Pr_{\mu} [\Good] \geq 1 - \delta / 4$,
such that for all $x \in \Good$, $S(\rho_x \| \rho) \leq 4 k/\delta$.
We will define superoperators $\cM_x$ for $x\in \Good$ and $x \not\in
\Good$ separately, and then show that they form a $(\delta,
\alpha)$-corrector.

Fix $x \in \Good$. Using Fact~\ref{fact:substate} with $r$ to be
chosen later, we conclude that for all $x \in \Good$, there is an
unnormalized superoperator $\tcM_x$ acting on $A$ only
such that if $q_x \defeq \Tr
\tcM_x(\ketbra{\phi})$, $\tsigma_x \defeq
\frac{\tcM_x(\ketbra{\phi})}{q_x}$ then, 
$q_x \geq \frac{r - 1}{r2^{4rk/\delta}}$
and $\trnorm{\tsigma_x - \ketbra{\phi_x}} \leq
\frac{2}{\sqrt{r}}$. Now, measure register $X$ in $\tsigma_x$ and
declare success if the result is $x$. Let ${\sigma'_x}$ be the
resulting normalized state when $x$ is observed.  
Measuring $X$ in $\ket{\phi_x}$ results gives the value $x$ with
probability $1$. Hence, by
Proposition~\ref{prop:trdistp},
\[\trnorm{\sigma'_x - \ketbra{\phi_x}} \leq \frac{2}{\sqrt{r}}.\]
Furthermore, since $\trnorm{\tsigma_x - \ketbra{\phi_x}} \leq
\frac{2}{\sqrt{r}}$, the probability $q'_x$ of observing $x$ when $X$ is
measured in the state $\tsigma_x$ is at least
$1-\frac{1}{\sqrt{r}}$, and the overall probability of success is at
least 
$q_x q'_x \geq (1-\frac{1}{\sqrt{r}})(\frac{r-1}{r2^{4rk/\delta}})
 \defeq \alpha$. 
In order to ensure that the overall probability of success is
exactly $\alpha$, we do a further {\em rejection} step:
Even on success we artificially declare
failure with probability $1 - \frac{\alpha}{q_x q'_x}$.
Let $\cM_x$ be the unnormalized superoperator which first
applies $\tcM_x$, then measures the register $X$, and on finding $x$
accepts with probability $\frac{\alpha}{q_x q'_x}$. Thus, for all $x
\in \Good$, the probability of success $r_x \defeq \Tr
\cM_x(\ketbra{\phi})$ is exactly equal to $\alpha$. This completes the
definition of $\cM_x$ for $x \in \Good$.

For $x \not \in \Good$, $\cM_x$ swaps $\ket{x}$ into register
$X$ from some outside ancilla initialized to $\ket{0}$
and declares success artificially with
probability  $r_x=\alpha$. For all $x \in \cX$, let
$\sigma'_x \defeq \frac{\cM_x(\ketbra{\phi})}{r_x}$.

Thus for all $x \in \cX$, $\sigma'_x$ contains $\ket{x}$ 
in register $X$ and $r_x = \alpha$. Finally, we have
\begin{eqnarray*}
\lefteqn{\E_{\mu} \trnorm{\sigma'_x - \ketbra{\phi_x}}} \\
& \leq & \sum_{x \in \Good} \mu(x) 
         \trnorm{\sigma'_x - \ketbra{\phi_x}} 
         + \sum_{x \not \in \Good} \mu(x) \cdot 2 \\
& \leq & \frac{2}{\sqrt{r}} + \frac{\delta}{4} \cdot 2.
\end{eqnarray*}
For $r=\frac{16}{\delta^2}$, this quantity is at most $\delta$, and we
conclude that the family $\{\cM_x\}_{x \in \cX}$ forms the
required $(\delta, \alpha)$-corrector for the ensemble
$\{\{\ket{\phi_x}\}_{x \in \cX}; \ket{\phi}\}$ with $\alpha =
2^{-O(k/\delta^3)}$.
\qed
\end{proof}

\subsection{Miscellaneous} We have the following Lemma.
\begin{lemma}
\label{lem:markovsubstate} Let $\delta > 0$. Let $P,Q$ be probability distributions with
support on set $\cX$ such that $S(P||Q) \leq c$. Then, we get a set
$\Good \subseteq \cX$ such that
\begin{equation}
\Pr_{P}[ x\in \Good] \geq  1 - \delta \quad \text{ and }  \quad  \forall
x \in \Good, \frac{P(x)}{Q(x)} \leq 2^{\frac{c+1}{\delta}}.
\end{equation}
\end{lemma}
\begin{proof}
We first have the following claim:
\begin{claim}Let $P$ and $Q$ be two distributions on the set
  $\cX$. For any set $\cX' \subseteq \cX$, we have
$$\sum_{x\in \cX'} P(x)\log
  \frac{P(x)}{Q(x)} \quad \geq \quad -\frac{\log e}e \quad > \quad -1.$$
\end{claim}
\begin{proof}
We require the following facts.
\begin{enumerate}
\item {\em log-sum inequality:} For non-negative integers
  $a_1,\dots,a_n$ and $b_1,\dots,b_n$, $$\sum a_i\log\frac{a_i}{b_i}
  \geq \left(\sum a_i\right) \log \frac{\sum a_i}{\sum b_i}.$$
\item The function $x\log x \geq -(\log e)/e$ for all $x >0$.
\end{enumerate}
From the above, we have the following sequence of inequalities.
\begin{eqnarray*}
\sum_{x\in \cX'} P(x)\log \frac{P(x)}{Q(x)}
& =& \sum_{x\in \cX'} P(x)\log \frac{P(x)}{Q(x)} +
\sum_{x\notin \cX'} Q(x)\log \frac{Q(x)}{Q(x)}\\
&\geq & \left(\sum_{x\in
    \cX'} P(x) +
  \sum_{x\notin\cX'}Q(x)\right)
    \log \left(\frac{\sum_{x\in
    \cX'} P(x) +
  \sum_{x\notin\cX'}Q(x)}{\sum_{x\in\cX}Q(x)}\right)\\
&=& \left(\sum_{x\in
    \cX'} P(x) +
  \sum_{x\notin\cX'}Q(x)\right)
    \log \left(\sum_{x\in
    \cX'} P(x) +
  \sum_{x\notin\cX'}Q(x)\right)\\
&\geq & -\frac{\log e}e
\end{eqnarray*}
\qed\end{proof}
Now:
\begin{eqnarray*}
c \geq S(P||Q) & = & \sum_{x: P(x) \geq Q(x)} P(x) \log \frac{P(x)}{Q(x)} +
\sum_{x: P(x) < Q(x)} P(x) \log \frac{P(x)}{Q(x)}  \\
& > & \sum_{x: P(x) \geq Q(x)} P(x) \log \frac{P(x)}{Q(x)} - 1 \\
\Rightarrow  c + 1 & > & \sum_{x: P(x) \geq Q(x)} P(x) \log
\frac{P(x)}{Q(x)} 
\end{eqnarray*}
Now we get our desired set $\Good$ immediately by using Markov's inequality.
\qed\end{proof}

We also need the following lemma.
\begin{lemma}
\label{lem:lowinfent}
Let $XAB$ be a tri-partite system with $X$ classical and $A,B$ quantum
systems. If $I(X:A)=0$ then $I(X:AB) \leq 2S(B)$.
\end{lemma}
\begin{proof}
We have the following Araki-Lieb~\cite{ArakiL70} inequality for any two systems $M_1,
M_2$: $|S(M_1) - S(M_2)| \leq S(M_1M_2)$. This implies:
$$ I(M_1:M_2) = S(M_1) + S(M_2) - S(M_1M_2) \leq \min\{2S(M_1), 2S(M_2)\}.$$  
Now,
\begin{eqnarray*} 
I(X:AB) & = & I(X:A) + I(XA:B) - I(A:B)  \\
& \leq &  I(XA:B) \leq 2S(B).
\end{eqnarray*}

\end{proof}

\section{One-way Message Compression and Optimal Direct Sum}
\label{sec:oneround}
Although in this section are concerned with message compression in
one-way protocols, we state our results in a general setting of
compressing the first message of multi-round two-way protocols. This
way of stating our message compression results is helpful in expressing our
round-elimination results. We state our results and proofs here only
for quantum protocols and the corresponding results for classical
protocols can be obtained in analogous fashion. We skip making explicit
statements and proofs for classical protocols for brevity.

\subsection{Message Compression and Round Elimination}  
We begin with the following definition.
\begin{definition}[\mbox{$[t; l_1, \ldots, l_t]^A$} protocol]
In a $[t; l_1, \ldots, l_t]^A$ protocol, there are 
$t$ rounds of communication with $\alice$
starting, the $i$th message being $l_i$ qubits long.
A $[t; l_1, \ldots, l_t]^B$
protocol is the same but $\bob$ starts the communication.
\end{definition}

\begin{theorem}[Compressing the first message]
\label{thm:oneround}
Let $f \subseteq \cX \times \cY \times \cZ$ be a function and $\mu$ be a
probability distribution on $\cX \times \cY$.  Suppose $\cP$ is a $[t;
l_1, l_2, \ldots, l_t]^A$ quantum protocol without prior entanglement
for $f$ having average error less than $\epsilon$ under $\mu$.  Let
$X$ denote the random variable corresponding to $\alice$'s input and
$N_1$ denote the qubits of $\alice$'s first message in $\cP$. Suppose
$I(X : N_1) \leq k$.  Let $\delta > 0$ be a sufficiently small
constant.  Then, there is a $[t; \beta, l_2 \ldots, l_t]^A$ quantum
protocol $\cP'$ with prior entanglement for $f$ with average error
less than $\epsilon + \delta$ under $\mu$, where $\beta =
O\left(\frac{k}{\delta^3}\right)$. Also, the first message of $\alice$ in
$\cP'$ is classical.  
\end{theorem}
\begin{proof}
Let $\ket{\phi_x}$ denote the state vector in $\cP$
of $\alice$'s qubits (including her input register) 
and her first message $N_1$
just after she sends $N_1$ to $\bob$, when she is given input 
$x \in \cX$. 
Let $\ket{\phi}$ denote the corresponding state vector
in $\cP$ when the protocol starts with $\alice$'s input registers
in the state $\sum_x \sqrt{p_x} \ket{\phi_x}$, where
$p_x \defeq \Pr_\mu [X = x]$.
Since $I(X : N_1) \leq k$, 
Lem.~\ref{lem:corrector} implies that 
there is a $(\delta/2, \alpha)$-corrector $\{\cM_x\}_{x \in \cX}$
for the
ensemble $\{\{\ket{\phi_x}\}_{x \in \cX}; \ket{\phi}\}$ where 
$\alpha=2^{-O(k/\delta^3)}$. 
That is, with $r_x \defeq \Tr (\cM_x \ketbra{\phi})$ and
$\sigma'_x \defeq 
 \frac{\cM_x (\ketbra{\phi})}{r_x}$, we have
$\E_\mu\left[\trnorm{\sigma'_x- \ketbra{\phi_x}}\right] 
 \leq \frac{\delta}{2}$.

We now describe the protocol $\cP'$. The protocol $\cP'$
starts with $2^\beta \defeq \alpha^{-1} \log (2/\delta)$ copies of 
$\ket{\phi}$ as prior entanglement. $\alice$
applies $\cM_x$ to each copy of $\ket{\phi}$ and sends the index of
the first copy on which she achieves success. Thus, her first message
in $\cP'$ is classical and 
$\beta = \log (1/\alpha) + \log \log (2/\delta) = O(k/\delta^3)$ 
bits long. $\alice$ and $\bob$ use 
that copy henceforth; the
rest of $\cP'$ is exactly as in $\cP$. The probability that $\alice$
achieves success with $\cM_x$ on at least one copy of $\ket{\phi}$ is
more than $1 - \frac{\delta}{2}$. Furthermore, the state of 
$\alice$'s registers
and the first message $N_1$ on this copy is exactly $\sigma'_x$.
Thus, the probability of error for the protocol $\cP'$ is at most
\[ 
\epsilon + \frac{\delta}{2} + \E_\mu
\left[\trnorm{\sigma'_x - \ketbra{\phi_x}}\right] 
\leq \epsilon + \frac{\delta}{2} + \frac{\delta}{2} \leq \epsilon +
\delta.
\]
This completes the proof of the theorem.
\qed\end{proof} 
\paragraph{Remark:}  
We can eliminate prior entanglement in quantum
protocols by assuming that
$\alice$ generates the prior entangled state herself, and then
sends $\bob$'s share of the state along with her first message.
This can make $\alice$'s first message long, but if 
the information about $X$ in
$\alice$'s first message together with $\bob$'s share of prior
entanglement qubits in the original
protocol is small, then the conclusions of the theorem still hold.
\begin{corollary}[Eliminating the first round]
\label{cor:oneround}
Under the conditions of Thm.~\ref{thm:oneround},
if $t \geq 3$ there is a 
$[t - 1; 2^{\beta} l_2, l_3 + \beta, l_4, \ldots, l_t]^B$ quantum 
protocol $\tcP$ with prior entanglement for $f$ with average 
error at most $\epsilon + \delta$ under $\mu$.  
If $t = 2$, we get a $[1; 2^{\beta} l_2]^B$
quantum protocol $\tcP$ with prior entanglement for $f$ with average
error at most $\epsilon + \delta$ under $\mu$.
\end{corollary}
\begin{proof}
Suppose $t \geq 3$.  Let $N_2$, $N_3$ denote the second and third
messages of $\cP'$.  Consider a 
$(t - 1)$-round protocol $\tcP$ where $\bob$ 
begins the communication by
sending his messages $N_2$ for all the $2^\beta$
copies of $\ket{\phi}$. This makes $\bob$'s first message in $\tcP$ to
be $2^\beta l_2$ qubits
long. $\alice$ replies by applying $\cM_x$ to each copy of $\ket{\phi}$
and sending the index of the first copy on which she achieves
success. She also sends her response $N_3$ corresponding to that copy
of $\ket{\phi}$. Thus, her first 
message in $\tcP$ is $l_3 + \beta$ qubits long. Note
that the operations of $\bob$ and the applications of $\cM_x$ by $\alice$ 
during the first two messages of
$\tcP$ are on disjoint sets of qubits, hence they commute.  Thus, the
global state vector of $\tcP$ after the second message is exactly the
same as the global state vector of $\cP'$ after the third
message. Hence the error probability remains the same. This proves the
first statement of the corollary.  The second statement of the 
corollary (case $t=2$) can be proved similarly.
\qed\end{proof}
\paragraph{Remark:} The above corollary can be thought of as 
the quantum analogue of the `message switching'
lemma of \cite{ChakrabartiR04}.

Using Corr.~\ref{cor:oneround}, we can now prove our new 
round elimination result for quantum protocols.
\begin{theorem}[Round elimination lemma]
\label{thm:roundelim}
Let $f: \cX \times \cY \rightarrow \cZ$ be a function and $k$, $t$ be
positive integers. Suppose $t \geq 3$.  Suppose $\cP$ is a $[t; l_1,
l_2, l_3, \ldots, l_t]^A$ quantum protocol with prior entanglement for
$f^{(k),A}$ (recall definition Def.~\ref{def:fk}) with worst case
error less than $\epsilon$. Let $\delta > 0$ be a sufficiently small
constant. Let $\beta \defeq O(\frac{l_1}{\delta^3 k})$. Then there is
a $[t - 1; 2^{\beta} l_2, l_3 + \beta, \ldots, l_t]^B$ quantum
protocol with prior entanglement for $f$ with worst case error at most
$\epsilon + \delta$.
\end{theorem} 
\begin{proof}({\bf Sketch})
The proof follows in a standard fashion by combining the proof
technique of Lem.~4 of \cite{sen:pred} with
Corr.~\ref{cor:oneround}. We skip making a complete proof for brevity.
\qed\end{proof} 
\paragraph{Remark:} The above round elimination
lemma is quantum analogue of a classical round elimination
result of Chakrabarti and Regev~\cite{ChakrabartiR04}. It allows us to
extend their optimal randomized cell probe lower bound for Approximate
Nearest Neighbor Searching in the Hamming cube $\{0, 1\}^n$ to the
quantum address-only cell probe model defined by Sen and
Venkatesh~\cite{sen:icalp01}. It also allows us to extend the sharper
lower bounds for Predecessor Searching of Patrascu and
Thorup~\cite{patrascu:pred} to the quantum case. We skip making
explicit statements and their proofs for brevity.

\subsection{One-Way Optimal Direct Sum} 
We get the following implication of Thm.~\ref{thm:oneround} to the
Direct Sum problem for one-round quantum communication
protocols. 
Recall that $f^{\oplus m}$ is the $m$-fold 
Direct Sum problem corresponding to the relation $f$.
\begin{theorem}[Direct Sum]
Let $f \subseteq \cX \times \cY \times \cZ$ be a relation. Let
$\epsilon, \delta \in (0,1/2)$ with $\epsilon + \delta < 1/2$. 
For one-round quantum protocols with prior entanglement, we get
$$\sQ^{1, A \rightarrow B, \pub}_\epsilon (f^{\oplus m}) 
\quad \geq \quad \Omega\left( \delta^3m \cdot \sQ^{1, A \rightarrow B,
\pub}_{\epsilon + \delta}(f)\right).$$
Similar result also holds by switching the roles of $\alice$ and $\bob$. 
For simultaneous message protocols, we get
$$\sQ^{\|, \pub}_\epsilon (f^{\oplus m}) \quad \geq \quad 
 \Omega\left(\delta^3m \cdot \sQ^{\|,\pub}_\epsilon (f)\right).$$
\end{theorem}
\begin{proof}
We present the proof for one-round protocols and the proof for $\smp$
protocols follows very similarly. Below we assume that in the one-way
protocols we consider the single message is from $\alice$ to $\bob$,
and hence we do not explicitly mention it in the superscripts. Let
$\epsilon, \delta$ be as in the statement of the theorem and let $c
\defeq \sQ^{1,\pub}_\epsilon (f^{\oplus m})$. 
For showing our result we will show that for all
distributions $\lambda$ on $\cX \times \cY$, 
\begin{equation} \label{eq:directavg}
\sQ^{1, \pub, \lambda}_{\epsilon + \delta}(f) \quad = \quad O(\frac{c}{\delta^3m}).
\end{equation}
Using Yao's principle and Eq.~(\ref{eq:directavg}), we immediately get
the desired result as follows:
\begin{equation*}
\sQ^{1, \pub}_{\epsilon + \delta} (f) \quad = \quad \max_{\lambda \text{
on }\cX \times \cY} \sQ^{1, \pub, \lambda}_{\epsilon + \delta} (f) \quad
= \quad O(\frac{c}{\delta^3m}) \quad = \quad O(\frac{1}{\delta^3m} \cdot
\sQ^{1, \pub}_\epsilon (f^{\oplus m})). 
\end{equation*}
Let us now turn to showing Eq.~(\ref{eq:directavg}). Since $ \sQ^{1,
\pub}_\epsilon (f^{\oplus m}) = c$, let $\cP$ be a
protocol (possibly using entanglement) for $f^{\oplus m}$ with
communication $c$ and error on every input being at most
$\epsilon$. Let us consider a distribution $\mu$ (possibly
non-product) on $\cX \times \cY$.
Our intention is to exhibit a protocol $\tcP$ for $f$ with
communication  $O(\frac{c}{\delta^3m})$ and distributional error at most
$\epsilon + \delta$ under $\mu$ and this would imply from definition that
$\sQ^{1,\pub, \mu}_{\epsilon + \delta}(f) = O( \frac{c}{\delta^3m})$,
and we would be done. 

From $\cP$ let us get a protocol $\cP'$ without prior entanglement in
which $\alice$ generates both parts of the shared state herself and then
sends $\bob$'s part as part of her first message. $\alice$ and $\bob$ then
behave identically as in $\cP$. Now let us provide inputs to $\cP'$ as
follows.  Let $\mu_X$ be the marginal of $\mu$ on $\cX$.  Recall that
$\alice$ has $m$ parts of the inputs in $\cP'$. Let the input of $\alice$ be
distributed according to $\mu_X$ in each part independently, and let
$\bob$ get input $0$ in every part.  Let $X$ be the random variable
representing the combined input of $\alice$.  Let $X_i, i \in [m]$ be the
random variable representing the input of $\alice$ on the $i$-th
co-ordinate. Note that $X_i, i \in [m]$ are all independent. Let $M$
represent the message of $\alice$. Now using Lem.~\ref{lem:lowinfent}
(irrespective of the number of qubits of prior entanglement in $\cP$)
we have $2c \geq I(X:M)$. Now from Chain Rule of mutual information we
get: $$ 2c \quad \geq \quad I(X: M) \quad = \quad
\sum_{i=1}^m I(X_i:M).$$ Therefore there exists a co-ordinate $i_0 \in
[m]$ such that $I(X_{i_0}:M) \leq
\frac{2c}{m}$. Now let us define a protocol $\cP''$ for $f$, in which 
on getting input $x \in \cX, y \in \cY$ respectively (sampled jointly
according to $\mu$), $\alice$ and $\bob$ simulate $\cP'$ by assuming $x$ and
$y$ to be inputs for the $i_0$-th co-ordinate. For the rest of the
co-ordinates $\alice$ generates the inputs independently according to the
distribution $\mu_X$. $\bob$ simply inserts $0$ as inputs in the rest of
the co-ordinates. $\alice$ then acts identically as in $\cP$ and sends
her message $M''$ to be $\bob$, who then outputs his decision as in $\cP$.
Note that in this case too $I(X''_{i_0}:M'') \leq \frac{2c}{m}$, where
$X''_{i_0}$ represents the input of $\alice$ in the $i_0$-th
co-ordinate. Also note that since the error of $\cP$ on every input
was at most $\epsilon$, we have that the distributional error under
$\mu$ in $\cP''$ is also at most $\epsilon$. 

We are now ready to define our intended protocol $\tcP$. Protocol
$\tcP$ is obtained by compressing the message of $\alice$ in protocol
$\cP''$ as according to Thm.~\ref{thm:oneround} (by assuming $t=1$).
Hence the message of $\alice$ in $\tcP$ has length
$O(\frac{c}{m\delta^3})$ and the distributional error of $\tcP$ under
$\mu$ is at most $\epsilon + \delta$. 
\qed\end{proof}

\section{Multi-Round Message Compression and Weak Direct Sum}
\label{sec:multiround}
\subsection{Quantum Protocols} 
In this section, we state and formally prove our results for
compressing messages in multi-round quantum communication protocols
for computing a relation $f \subseteq \cX \times \cY \times \cZ$. In our
discussion below, $A, X, B, Y$ denote $\alice$'s work qubits, $\alice$'s
input qubits, $\bob$'s work qubits and $\bob$'s input qubits respectively,
at a particular point in time.
\begin{definition}[Privacy loss]
\label{def:privacyquant}
Let $\mu \defeq \mu_\cX \times \mu_\cY$ be a product probability 
distribution on $\cX \times \cY$. 
Suppose $\cP$ is a quantum protocol for a relation
$f \subseteq \cX \times \cY \times \cZ$.
Consider runs of $\cP$ when $\alice$'s input register $X$
starts in the mixed state $\sum_{x \in \cX} \mu_\cX(x) \ketbra{x}$ 
and $\bob$'s input
register $Y$ starts in the pure state 
$\sum_{y \in \cY} \sqrt{\mu_\cY(y)} \ket{y}$. Let $B$ denote
the qubits in the possession of $\bob$ including $Y$, at some point
during the execution of $\cP$. Let $I(X : B)$ denote the
mutual information of $\alice$'s input register $X$ with $\bob$'s qubits
$B$. The {\em privacy loss}
of $\cP$ for relation $f$ on the distribution $\mu$ from $\alice$
to $\bob$ at that point in time
is $L^{\cP}(f, \mu, A, B) \defeq I(X : B)$.  The privacy loss
from $\bob$ to $\alice$, $L^{\cP}(f, \mu, B, A)$, is defined similarly. 
The privacy loss of $\cP$ from $\alice$ to $\bob$ for $f$, 
$L^{\cP}(f, A, B)$,
is the 
maximum over all product distributions $\mu$ of 
$L^{\cP}(f, \mu, A, B)$.
The privacy loss of $\cP$ from $\bob$ to $\alice$ for $f$, 
$L^{\cP}(f, B, A)$, is defined similarly.
The privacy loss from $\alice$ to $\bob$ for $f$, $L(f, A, B)$,
is the infimum over all protocols $\cP$ of $L^{\cP}(f, A, B)$ at
the end of $\cP$. The quantity
$L(f, B, A)$ is defined similarly. 
\end{definition}

\begin{theorem}[Compressing many rounds]
\label{thm:multiround}
Suppose $\cP$ is a $[t;l_1,l_2,\ldots,l_t]^A$ quantum protocol without
prior entanglement for a relation $f \subseteq \cX \times \cY \times \cZ$.
Let $\mu \defeq \mu_\cX \times \mu_\cY$ be a product probability
distribution on $\cX \times \cY$.  Suppose the average error of $\cP$
when the inputs are chosen according to $\mu$ is at most $\epsilon$.
Let $k_a$, $k_b$ denote the privacy losses of $\alice$ and $\bob$
respectively after $t'$ rounds of communication. Suppose $t'$ is odd
(similar statements hold for even $t$, as well as for interchanging
the roles of $\alice$ and $\bob$).  Then, for all sufficiently small
constants $\delta > 0$, there exists a $[t-t'+1; \lambda_1, \lambda_2,
l_{t'+2}, \ldots, l_{t}]^A$ protocol $\cP'$ in the presence of prior
entanglement such that:
\begin{enumerate}
\item the average error of $\cP'$ with respect to $\mu$
      is at most $\epsilon+\delta$;
\item $\lambda_1 \leq k_a \cdot 2^{O(k_b/\delta^{6})}$ 
      and $\lambda_2 \leq l_{t'+1} + O(k_b/\delta^6)$.
\end{enumerate}
\end{theorem}
\begin{proof}
Consider the situation after $t'$ rounds of
$\cP$. Let the joint state of $\alice$ and $\bob$ be denoted by
\begin{description}
\item 
{${\sigma_{xy}}$:}  when $\alice$ starts $\cP$ with $x$ in her input
register and $\bob$ starts with $y$ in his input register;
\item
{${\sigma_x}$:} when $\alice$ starts
with $x$ in her input register and $\bob$ starts with the superposition
$\sum_{y \in \cY} \sqrt{\mu_\cY(y)} \ket{y}$ in his input register; 
\item
{${\sigma_y}$:} when $\bob$ starts
with $y$ in his input register and $\alice$ starts with the 
superposition $\sum_{x \in \cX} \sqrt{\mu_\cX(x)} \ket{x}$ in her 
input register;
\item
{${\sigma}$:} when $\alice$ and $\bob$ start with the superposition 
$\sum_{(x,y) \in \cX \times \cY} \sqrt{\mu(xy)} \ket{x}\ket{y}$ in 
their input registers.
\end{description}
Note that $\sigma_{xy}$, $\sigma_x$, $\sigma_y$ and $\sigma$ are
pure states.

We overload the symbols $\cX, \cY$ to also denote
the superoperators corresponding to measuring in the computational
basis the input
registers $X, Y$ of $\alice$ and $\bob$ respectively.
Whether $\cX, \cY$ denote sets or superoperators will be clear
from the context. When several
superoperators are applied to a state in succession we omit the
parenthesis; for example, we write $\cX\cY(\rho)$ instead of 
$\cX(\cY(\rho))$ which corresponds to measuring the input registers of
$\alice$ and $\bob$ (in this case, their order does not matter).

We will choose $\delta_a, \delta_b > 0$ later.
Since the privacy loss
of $\alice$ is at most $k_a$, Lem.~\ref{lem:corrector} implies that
there is a $(\delta_a,\alpha)$-corrector $\{\cM_x\}_{x \in \cX}$ for
$\{\{\sigma_x\}_{x \in \cX}; \sigma\}$ with 
$\alpha = 2^{-O(k_a/\delta_a^3)}$. 
Similarly, since the privacy loss of $\bob$
is at most $k_b$, there is a $(\delta_b,\beta)$-corrector 
$\{\cM_y\}_{y \in \cY}$ for $\{\{\sigma_y\}_{y \in \cY}; \sigma\}$ 
with $\beta=2^{-O(k_b/\delta_b^3)}$. In
particular, with $\cM_X \defeq \E_{\mu_\cX} [\cM_x]$ and 
$\cM_Y \defeq \E_{\mu_\cY} [\cM_y]$, we have
\begin{equation}
\begin{array}{l}
\trnorm{\frac{\cM_X(\sigma)}{\alpha} - \cX(\sigma)} \leq \delta_a,   \\ \\
\trnorm{\frac{\cM_Y(\sigma)}{\beta} - \cY(\sigma)} \leq \delta_b. 
\end{array}
\label{eq:corrector}
\end{equation}
In our proof, we will take
\begin{equation}
\delta_b \defeq \left(\frac{\delta}{10}\right)^2, ~~~
\delta_a \defeq \frac{\delta_b \beta}{2}.
\label{eq:deltas}
\end{equation}

The proof has two steps. In the first step, we analyze the protocol
$\cP'$ given in Figure~\ref{fig:zeroround}.  In $\cP'$, $\alice$
and $\bob$ try to recreate the effect of the first $t'$ rounds of the
original protocol, but without sending any messages. For this, they
start from the state $\sigma$ (their prior entanglement) and on
receiving $x$ and $y$, apply suitable correcting transformations.  In
the second step, we shall consider a protocol $\cP''$ that starts 
with several parallel executions of $\cP'$.

\begin{figure}[ht] 
\begin{center}
\mybox{
\smallskip
{\em $\alice$ and $\bob$ start with the joint state 
$\sigma$ as prior entanglement.}
\begin{description}
\item[Input:] $\alice$ is given $x\in X$; $\bob$ is given $y\in Y$. 
\item[$\alice$:] Applies superoperator $\cM_x$ to her registers. 
\item[$\bob$:]   Applies superoperator $\cM_y$ to his registers.
\end{description}
\smallskip
}
\caption{The intermediate protocol $\cP'$}
\label{fig:zeroround}
\end{center}
\end{figure}

Let $r_{xy}\defeq\Tr \cM_y\cM_x(\sigma)$ and let 
$r \defeq \E_\mu[r_{xy}]$. 
Then, $r_{xy}$ is the probability that both $\alice$ and $\bob$ succeed
on input $(x,y)$, and $r$ is the probability that they
succeed when their input is chosen according to the distribution
$\mu$. Let $\rho$ denote the state after $t'$ rounds
of $\cP$ when the inputs are chosen according
to $\mu$ i.e. $\rho \defeq \E_\mu[\sigma_{xy}]$. Observe
that $\rho = \cY \cX (\sigma)$.
Let $\rho'$ be the state at the end of $\cP'$, when the
inputs are chosen according to $\mu$ and we condition on both parties
succeeding i.e. $\rho'  = \frac{\cM_Y \cM_X (\sigma)}{r}$.
\begin{claim} 
\label{cl:zerocomm}
\begin{enumerate}
\item[(a)] $1-\frac{\delta_b}{2} \leq \frac{r}{\alpha\beta} \leq
1+\frac{\delta_b}{2}$.
\item[(b)] $\trnorm{\rho - \rho'} \leq 2\delta_b$.
\item[(c)] $\Pr_\mu\left[ \left|\frac{r_{xy}}{r} - 1\right| \geq 
            2\delta_b^{1/2}\right] \leq \delta_b^{1/2} $.
\end{enumerate}
\end{claim}
\begin{proof}
\begin{enumerate}
\item[(a)] 
\begin{eqnarray*} 
\frac{r}{\alpha\beta} 
&=& \frac{\Tr \cM_Y\cM_X (\sigma)}{\alpha\beta} \\
&=& \frac{1}{\beta} \Tr \left(\cM_Y
\left(\frac{\cM_X(\sigma)}{\alpha}\right)\right)\\
&=& \frac{1}{\beta}\Tr \cM_Y \cX(\sigma) + 
 \frac{1}{\beta} \Tr \cM_Y
    \left(\frac{\cM_x(\sigma)}{\alpha} - \cX(\sigma)\right).
\end{eqnarray*}
The first term on the right is $1$ since $\cM_\cY$ and $\cX$ commute
as they act on disjoint sets of qubits. For the second term, we have
using (\ref{eq:corrector}), (\ref{eq:deltas}) and the fact that an
unnormalized superoperator cannot increase the trace norm, that
\[ 
\left| \frac{1}{\beta}\Tr
\cM_Y\left(\frac{\cM_X(\sigma)}{\alpha} - \cX(\sigma) \right)\right| 
\leq
\frac{\delta_a}{\beta} 
= \frac{\delta_b}{2}.
\]

\item[(b)] 
Using (\ref{eq:corrector}), (\ref{eq:deltas}), the
fact that a measurement or an
unnormalized superoperator cannot increase the trace norm, and
that $\cM_\cY$ and $\cX$ commute 
as they act on disjoint sets of qubits, we get
\begin{eqnarray*}
\trnorm{\rho-\rho'} 
\trnorm{\cX\cY(\sigma) - \rho'} 
&\leq & \trnorm{\cX\frac{\cM_\cY(\sigma)}{\beta} -\rho'} +  
	 \trnorm{\cX\left(\cY(\sigma) - \frac{\cM_\cY(\sigma)}{\beta}
                    \right)} \\
&\leq & \trnorm{\cM_\cY\frac{\cX(\sigma)}{\beta} -\rho'} + \delta_b\\
&\leq & \trnorm{\frac{1}{\beta} \cM_Y \frac{\cM_X(\sigma)}{\alpha} -
               \rho'} + \delta_b +  \frac{1}{\beta}
       \trnorm{\cM_Y\left(\cX(\sigma) - \frac{\cM_X(\sigma)}{\alpha}
               \right)} \\
&\leq & \trnorm{ \frac{1}{\beta} \cM_Y \frac{\cM_X(\sigma)}{\alpha} -
  \rho'} + \delta_b + \frac{\delta_a}{\beta} \\
&\leq & \trnorm{\frac{r}{\alpha\beta} \frac{\cM_Y\cM_X(\sigma)}{r}
               -\rho'} + \frac{3 \delta_b}{2} \\
&=& \trnorm{\left(\frac{r}{\alpha\beta}-1\right) \rho'} + 
    \frac{3 \delta_b}{2}\\
&\leq & 2\delta_b. 
\end{eqnarray*}

\item[(c)] Let $\tau$ describe the joint state of the input 
registers when
the combined state of $\alice$ and $\bob$ is $\rho$; similarly, let $\tau'$
be the state of their input registers when the combined state is
$\rho'$; thus,
\[
\tau = \sum_{xy} p_{xy} \ketbra{x} \otimes
\ketbra{y} 
\]
and
\[
\tau' = \sum_{xy} p_{xy} \frac{r_{xy}}{r} \ketbra{x}\otimes\ketbra{y}.
\]
Using part (b), we have 
\[ \sum_{xy} p_{xy} \left| 1 - \frac{r_{xy}}{r}\right| =
\trnorm{\tau-\tau'} \leq \trnorm{\rho - \rho'} \leq 2\delta_b.\]
Thus, 
$\E_\mu\left[\left|\frac{r_{xy}}{r} - 1\right|\right] \leq 2\delta_b$,
and by Markov's inequality, $\Pr_\mu\left[ \left|\frac{r_{xy}}{r} -
1\right| \geq 2\delta_b^{1/2}\right] \leq \delta_b^{1/2}$.
\end{enumerate}
\unskip
\qed\end{proof}

We can now move to the second step of our proof of 
Thm.~\ref{thm:multiround}.
Figure~\ref{fig:final} presents a protocol $\cP''$
with $t-t'+1$ rounds of
communication where the initial actions of $\alice$ and $\bob$ are derived
from the protocol $\cP'$ analyzed above.
\begin{figure}[ht]
\begin{center}
\mybox{
\smallskip
{\em $\alice$ and $\bob$ start with 
     $K \defeq \frac{10}{r} (\log \frac{1}{\delta})$ copies of 
     $\sigma$ as prior entanglement. We refer to these copies as 
     $\sigma^1,\ldots, \sigma^K$.}
\begin{description}
\item[Input:] $\alice$ gets $x \in X$ and $\bob$ gets $y \in Y$.
\item[$\alice$:] Applies $\cM_x$ to each $\sigma^i$. Let
              $\hat{S}\defeq\{i: \mbox{$\cM_x$ succeeded on 
                                 $\sigma^i$}\}$.
              If $\hat{S}$ has less than $2\alpha K$ elements, $\alice$ 
              aborts the protocol; otherwise, she sends 
              $S \subseteq \hat{S}$ to $\bob$, $|S| = 2 \alpha K$.
\item[$\bob$:]   Applies $\cM_y$ to each $\sigma_i$ for $i\in S$ 
              and sends $\alice$ the index $i^*$ where he (and hence 
              both) succeeded. If there is is no such $i^*$ he 
              aborts the protocol.
\end{description}
$\alice$ and $\bob$ now revert to protocol $\cP$ after 
round $t'$,
and operate on the registers corresponding to $\sigma^{i^*}$.
\smallskip
}
\end{center}
\caption{The final protocol $\cP''$} 
\label{fig:final}
\end{figure}

\begin{claim}
\begin{enumerate}

\item[(a)] The number of bits sent by $\alice$ in the first round
is at most $k_a 2^{O(k_b/\delta^{6})}$; the number of bits
sent by $\bob$ is at most $O(k_b/\delta^{6})$.

\item[(b)] If the inputs are chosen according to the distribution
$\mu$, the protocol $\cP''$ computes $f$ correctly with probability of
error at most $\epsilon+ \delta$.
\end{enumerate}
\end{claim}
\begin{proof}
Recall that $\delta_b = (\delta/10)^2$, 
$\beta= 2^{-O(k_b/\delta_b^3)}$
and $\delta_a= \delta_b \beta/2$ and $\alpha=2^{-O(k_a/\delta_a^3)}$.
By part~(a) of Claim~\ref{cl:zerocomm}
it follows that 
$r\geq \alpha \beta/2$. 
The number of bits needed by $\alice$ to encode her set $S$ is at most
\[ 
\log {K \choose {2\alpha K}} \leq  
2 \alpha K \log\left(\frac{e}{2 \alpha}\right) = 
k_a 2^{O(k_b/\delta^6)}.
\]
The number of bits sent by $\bob$ is at most 
$\log 2\alpha K = O\left(\frac{k_b}{\delta^6}\right)$.
This justifies part~(a) of our claim.

For part (b), we will use Claim~\ref{cl:zerocomm} to 
bound the probability
of error $\cP''$.  Call a pair $(x,y) \in \cX \times \cY$ 
{\em good} if
$|\frac{r_{xy}}{r} -1| \leq 2\delta_b^{1/2}$; let $\chi$
denote the indicator random variable for the event ``$(x,y)$ is
good.'' Let $\chi'$ be the indicator random variable for the event
``$\alice$ and $\bob$ do not abort protocol $\cP''$.'' Note that if 
$\alice$ and
$\bob$ do not abort protocol $\cP''$, they enter round $t'+1$ 
of protocol $\cP$ with their
registers in the state
$\sigma'_{xy} \defeq \frac{\cM_x\cM_y(\sigma)}{r_{xy}}$.  Thus under
distribution $\mu$, the
average probability of error of $\cP''$ 
differs from the average
probability of error $\epsilon$ of the original
protocol $\cP$ by at most
\begin{equation}
\begin{array}{l}
\E_\mu\left[\chi \chi' \trnorm{\sigma'_{xy} - \sigma_{xy}}\right] + 
\Pr[\chi=0] + \Pr[\chi=1 \mbox{ and }\chi'=0].
\end{array}
\label{eq:bound:infty}
\end{equation}
The first term in the above sum can be bounded as follows:
\begin{eqnarray*}
\E_\mu\left[\chi \chi' \trnorm{\sigma'_{xy} - \sigma_{xy}}
               \right] &=& 
\E_\mu\left[\chi \chi' \trnorm{\frac{1}{r_{xy}}\cM_x\cM_y(\sigma) - 
\sigma_{xy}}\right]\\
&\leq& \E_\mu\left[\chi \chi' \trnorm{\frac{1}{r}\cM_x\cM_y(\sigma) - 
       \sigma_{xy}}\right] + \\
& &
\E_\mu\left[\chi \chi' \left|1-\frac{r_{xy}}{r}\right| 
       \frac{1}{r_{xy}} 
\trnorm{\cM_x\cM_y(\sigma)}\right]\\
& \leq & 
\trnorm{\frac{1}{r} \cM_Y \cM_X (\sigma) - \cX\cY (\sigma)} + \\
& &
\E_\mu\left[\chi \chi' \left|1-\frac{r_{xy}}{r}\right| 
       \frac{1}{r_{xy}} 
\trnorm{\cM_x\cM_y(\sigma)}\right]\\
&\leq& 2\delta_b + 2\delta_b^{1/2}.
\end{eqnarray*}
For the second last inequality, we used the fact that in the states
$\sigma'_{xy}$ and $\sigma_{xy}$, the input registers of $\alice$ and $\bob$
contain $x$ and $y$. For the last inequality, we used part~(b) of 
Claim~\ref{cl:zerocomm} and the definition of good $(x, y)$. 
The second term of
(\ref{eq:bound:infty}) is at most $\delta_b^{1/2}$ by part~(c) of 
Claim~\ref{cl:zerocomm}. It remains to bound the last term of
(\ref{eq:bound:infty}), which corresponds to the probability that
$\alice$ or $\bob$ abort the protocol for some good $(x,y)$.
\begin{description}
\item[$\alice$ aborts:] The probability of success of $\cM_x$
for any one copy of
$\sigma$ is exactly $\alpha$. Thus, the expected number of successes
is $\alpha K$, and by Chernoff's bound 
(see e.g.~\cite[Appendix A]{AlSp}), 
the probability that there are less than $2\alpha K$ 
successes is at most 
$\left(\frac{e}{4}\right)^{\alpha K} \leq \delta^{10}$.

\item[$\bob$ aborts:] $\bob$ aborts when the two parties do not 
simultaneously succeed in any of the $K$ attempts, even though their
probability of success was at least $r_{xy} \geq (1-2\delta_b^{1/2})
r\geq r/2$ (recall that we are now considering a good pair $(x,y)$).
The probability of this is at most
$\left( 1- \frac{r}{2}\right)^K \leq \exp\left(-\frac{rK}{2}\right) 
  \leq \delta^5$.
\end{description}
Thus overall, the average probability of error of $\cP''$ 
is at most 
\[ 
\epsilon + 2\delta_b + 2 \delta_b^{1/2} + \delta_b^{1/2} +
\delta^{10}+ \delta^5 
\leq \epsilon + \delta.
\]
\qed \end{proof}

This completes the proof of Thm.~\ref{thm:multiround}.
\qed \end{proof}

The following corollaries result from the above theorem.
\begin{corollary}[Privacy tradeoff]
For any relation $f \subseteq \cX \times \cY \times \cZ$, 
$L(f, A, B) 2^{O(L (f, B, A))} \geq 
 \sQ^{1, A \rightarrow B, {\pub}, [\;]} (f)$. Similarly,
$L (f, B, A) 2^{O(L (f, A, B))} \geq 
 \sQ^{1, B \rightarrow A, {\pub}, [\;]} (f)$.
\end{corollary}
\paragraph{Remark:} It was shown by Kremer~\cite{kremer:quantcc} 
that $Q(f) \geq
\Omega(\log D^{1} (f))$, where $D^{1} (f)$ is the one-round
deterministic communication complexity of $f$. The above corollary can
be viewed as the privacy analogue of that result. It is optimal as
evidenced by the Index function problem and the Pointer Chasing
problem, both of which have communication complexity 
$O(\log n)$~\cite{jain:substate}.
\begin{corollary}[Weak Direct Sum]
\label{corr:weakdirectquant} For any relation $f \subseteq \cX \times \cY \times \cZ$, 
$$\sQ^{{\pub}, [\;]} (f^{\oplus m}) \quad \geq \quad m \cdot
\Omega(\log \sQ^{1, {\pub}, [\;]} (f)).$$
\end{corollary}
\paragraph{Remark:} 
Jain, Radhakrishnan, and Sen~\cite{jain:icalp03,HarshaJMR07} proved  
Direct Sum results for classical multi-round protocols. Their results were stronger
because it avoided the logarithm. However, if we want a Direct
Sum result independent of the number of rounds, the above is the
best possible as evidenced by the Index function problem and the
Pointer Chasing problem~\cite{jain:substate}. 


\subsection{Classical Protocols}

Let $\cP$ be a classical private-coins two-way protocol for a relation
$f \subseteq \cX \times \cY
\times \cZ$.  Let $\mu_{{X}}, \mu_{{Y}}$ be probability
distributions on $ \cX, \cY$, and let $\mu \defeq \mu_{{X}} \times
\mu_{{Y}}$ denote a product distribution on
$\cX \times \cY$. Consider a run of $\cP$, in which the inputs of
$\alice$ and $\bob$, are drawn according to distribution
$\mu$. Let $X$ and $Y$ denote the random variables corresponding to
the input of $\alice$ and $\bob$ respectively. Let $M$ denote the complete
transcript of the messages sent by $\alice$ and $\bob$ during the protocol.
Let $I(X :M)$ denote the mutual information between random variables
$X$ and $M$ at the end of this run of ${\cal P}$.

\begin{definition}[Privacy loss]
\label{def:privacyclass}
The {\em privacy loss} of ${\cal P}$ for relation $f$ on the 
product distribution $\mu$ from $\alice$
to $\bob$ is defined as $L^{{\cal P}}(f, \mu, A, B) \defeq I(X : M)$.
The privacy loss from $\bob$ to $\alice$, 
is defined similarly as $L^{{\cal P}}(f, \mu, B, A) \defeq I(Y : M)$.
\end{definition}

\begin{theorem}
\label{thm:cltworound}
Let $f \subseteq \cX \times \cY \times \cZ$ be a relation and let
$\epsilon \in (0,1/2)$. Let $\mu$ be a product distribution on $\cX
\times \cY$. Let $\cP$ be a private-coins protocol for $f$ with
distributional error at most $\epsilon$ under $\mu$. Let us assume
without loss of generality that $\alice$ sends the first message and $\bob$
computes the final answer. Let $L^{{\cal P}}(f, \mu, A, B) \leq k_a$
and $L^{{\cal P}}(f, \mu, B, A) \leq k_b$.  Let $\tilde{\delta} >0$ be
such that $\epsilon + \tilde{\delta} \in (0,1/2)$. Then there exists a
one-round public-coin protocol (and hence also a deterministic
protocol) $\tilde{\cP}$ with single communication from $\alice$, such
that,
\begin{enumerate}
\item
Communication from $\alice$ in $\tilde{\cP}$ is
$O\left(\frac{\log \frac{1}{\tilde{\delta}}}{\tilde{\delta}^3}\cdot (k_a + 1) \cdot
2^{O((k_b+1)/\tilde{\delta}^2)}\right)$. 
\item  The distributional error of $\tcP$ under $\mu$ is at most
$\epsilon + \tilde{\delta}$.
\end{enumerate}
\end{theorem}
\begin{proof}
Let the marginals of $\mu$ on $\cX, \cY$ be $\mu_X,\mu_Y$
respectively. Therefore $\mu = \mu_X \otimes \mu_Y$. Let the
distribution of $M$ (the combined message transcript in $\cP$), when
$X=x$ and $Y=y$, be ${P_{x,y}}$.  Let $P_x
\defeq \av_{y \leftarrow \mu_Y}[P_{x,y}]$, $P_y \defeq \av_{x
\leftarrow \mu_X}[P_{x,y}]$  and $P \defeq \av_{(x,y) \leftarrow \mu}[P_{x,y}]$. 
Let there be $k$ messages in protocol $\cP$. Let $M_1, M_2, \ldots
M_k$ denote the random variables corresponding to the first, second
and so on till the $k$-th message of the protocol $\cP$. Let $S$ be
the set of all message strings $s$. For $s \in S$,
let $s_1, s_2, \ldots, s_k$ denote the parts corresponding to $M_1,
M_2, \ldots M_k$ respectively. For $i
\in [k]$, let $p^{x,y}(s, i)$ denote the probability with which
$s_i$ appears in $P_{x,y}$ conditioned on the first $i-1$ messages as
being $s_1, s_2, \ldots s_{i-1}$. Similarly we define $p^{x}(s, i)$,
$p^{y}(s, i)$ and $p(s, i)$ corresponding to distributions $P_x$,
${P_y}$ and $P$. Let $p^{x,y}(s)$ denote the probability with which
message $s$ appears in $P_{x,y}$. Similarly let us define $p^{x}(s)$,
$p^{y}(s)$ and $p(s)$ corresponding to distributions $P_x$, ${P_y}$
and $P$. Now we have the following claim.
\begin{claim}
\label{claim:nice}
For all $x \in \cX, y \in \cY, s \in S$,
$$p^x(s) \cdot p^y(s)  \quad = \quad p(s) \cdot p^{x,y}(s).$$ 
\end{claim}
\begin{proof}
Note that since $\cP$ is a private coins protocol and $\bob$ sends even
numbered messages, we have for all even $i, \forall x \in \cX, \forall
s \in S, p^x(s,i) = p(s,i)$. Therefore $\forall x \in \cX, \forall s
\in S$, 
\begin{equation} \label{eq:1}
\frac{p^x(s)}{p(s)} =
\frac{\prod_{i =1}^k p^x(s,i)}{\prod_{i =1}^k p(s, i)} =
\frac{\prod_{i :\mathsf{odd}}p^x(s,i)}{\prod_{i: \mathsf{odd}}p(s,
i)}.
\end{equation}
 Similarly we
have for all odd $i, \forall y \in \cY, \forall
s \in S, p^y(s,i) = p(s,i)$ and hence, 
\begin{equation} \label{eq:2}
 \frac{p^y(s)}{p(s)}
= \frac{ \prod_{i : \mathsf{even}}p^y(s,i)}{\prod_{i:
\mathsf{even}}p(s, i)}.
\end{equation}
We can note further that for $\forall x \in \cX, \forall
y \in \cY, \forall s \in S$; for all odd $i, p^{x,y}(s,i)= p^x(s,i)$ and for
all even $i, p^{x,y}(s,i)= p^y(s,i)$. Therefore, 
\begin{equation} \label{eq:3}
p^{x,y}(s) =
\prod_{i=1}^k p^{x,y}(s,i) = \prod_{i : \mathsf{odd}}p^x(s,i) \cdot
\prod_{i : \mathsf{even}}p^y(s,i).
\end{equation}
 Our claim now follows by combining Eq.~(\ref{eq:1}), Eq.~(\ref{eq:2})
and Eq.~(\ref{eq:3}).
\qed\end{proof}
Let $\delta = \frac{\tilde{\delta}}{5}$. Since $ k_a \geq I(M:X) = \E_{x
\leftarrow \mu_X}[S(P_x|| P)]$, using Markov's inequality we get a set
$\Good_X \subseteq \cX$ such that
\begin{equation}\label{eq:goodX}
 \Pr_{\mu_X}[ x\in \Good_X] \geq 1 - \delta \quad \text{ and } \quad \forall x \in
\Good_X, S(P_{x} || P) \leq \frac{k_a}{\delta}.    
\end{equation}
Let $x \in \Good_X$.  Since $ \frac{k_a}{\delta} \geq S(P_x || P) =
\E_{s \leftarrow P_x} \left[\log \frac{p^x(s)}{p(s)}\right]$, using
Lem.~\ref{lem:markovsubstate}, we get a set $\Good^x \subseteq S$ such
that
\begin{equation}\label{eq:goodx}
\Pr_{P_x}[ s\in \Good^x] \geq  1 - \delta \quad \text{ and }  \quad  \forall
s \in \Good^x, \frac{p^x(s)}{p(s)} \leq 2^{\frac{k_a+1}{\delta^2}}.
\end{equation}
Similarly there exists a set $\Good_Y \subseteq \cY$ such that
\begin{equation}\label{eq:goodY}
 \Pr_{\mu_Y}[ y\in \Good_Y] \geq 1 - \delta \quad \text{ and } \quad 
  \forall y \in \Good_Y,    S(P_{y} || P) \leq \frac{k_b}{\delta}.
\end{equation}
Similarly for $y \in \Good_Y$, there exists a set $\Good^y \subseteq S$
such that
\begin{equation}\label{eq:goody}
\Pr_{P_y}[ s\in \Good^y] \geq  1 - \delta \quad \text{ and }
\quad  \forall s \in \Good^y,  \frac{p^y(s)}{p(s)} \leq 2^{\frac{k_b+1}{\delta^2}}.
\end{equation}

Let us now present an intermediate protocol $\cP'$
in Fig.~\ref{fig:intclass} from which we will finally obtain our
desired protocol $\tcP$. 

\begin{figure}[ht]
\begin{center}
\mybox{
\smallskip
{\em $\alice$ and $\bob$, using shared prior randomness, generate an array
of strings (each string belonging to the set $S$) with infinite columns
and $K \defeq \left(\frac{1}{1-\delta} \cdot 
\ln \frac{1}{\delta}\right) \cdot 2^{(k_b+1)/\delta^2}$
rows. Each string in the array is sampled independently according to
the distribution $P$. Let the random variables representing various
strings be $S^{i,j}, i \in [K], j \in \mathbb{N}$ ($\mathbb{N}$ is the
set of natural numbers). }
     
\begin{description}
\item[Input:] $\alice$ gets $x \in X$ and $\bob$ gets $y \in Y$. 

\item[$\alice$:] She sets $i=1, j=1$.
		\begin{enumerate} \item In case $x \notin \Good_X$,
		she aborts the protocol and sends a special abort
		message to $\bob$ (using constant number of bits). Otherwise
		she moves to step 2.  \item She considers string
		$S^{i,j}$. In case $S^{i,j} \in \Good_x$, she accepts
		$S^{i,j}$ with probability $\frac{1}{2^{(k_a
		+1)/\delta^2}} \cdot
		\frac{p^x(S^{i,j})}{p(S^{i,j})}$. In case $S^{i,j}
		\notin \Good_x$, she accepts $S^{i,j}$ with
		probability $0$. \item In case she accepts $S^{i,j}$,
		she communicates $j$ to $\bob$ using a prefix free binary
		encoding. If $i = K$, she stops, otherwise she sets $i
		= i +1, j=1$ and goes to step 2. In case she rejects
		$S^{i,j}$, she sets $j=j+1$ and moves to step 2.
		\end{enumerate} Let the various index communicated to
		$\bob$ be denoted $J_i, i \in [K]$.
\item[$\bob$:] 	He sets $l=1$. If he gets abort message from $\alice$, he
		aborts the protocol, otherwise he goes to step 1.
		\begin{enumerate} \item If $y \notin \Good_y$, he
		aborts the protocol. Otherwise he goes to step 2.
		\item He considers the string $S^{l,J_l}$, where $J_l$
		is as obtained from $\alice$. If $S^{l,J_l} \in \Good_y$,
		he accepts $S^{l,J_l}$ with probability
		$\frac{1}{2^{(k_b+1)/\delta^2}} \cdot
		\frac{p^y(S^{l,J_l})}{p(S^{l,J_l})}$. If $S^{l,J_l} \notin \Good_y$,
		he accepts $S^{l,J_l}$ with probability $0$. \item In
		case he accepts $S^{l,J_l}$, he considers it to be the
		the final message transcript $M$ of protocol $\cP$ and simulates
		$\cP$ from now on to output $z \in \cZ$. In case he
		rejects $S^{l,J_l}$, if $l=K$ he aborts the protocol,
		otherwise he sets $l = l+1$ and goes to step 2.
		\end{enumerate}
\end{description}
\smallskip
}
\end{center}
\caption{The intermediate protocol $\cP'$} 
\label{fig:intclass}
\end{figure}

Protocol $\cP'$ is clearly one-way protocol. Now let us now analyze
the expected communication from $\alice$ to $\bob$ in $\cP'$ and expected
error of $\cP'$. 

\smallskip
\noindent{\bf Expected communication of $\cP'$:} 
 When $x \notin \Good_X$, there is constant communication. Let $x \in
 \Good_X$, and fix $i \in [K]$. Then the probability that $J_i=j$ 
 given that the previous samples were rejected in the row $i$, is:
\begin{eqnarray*} 
\lefteqn{\sum_{s \in S} \Pr(S^{i,j} =s) \cdot  \Pr(s \text{ is accepted})} \\ 
& = & \sum_{s \in \Good_x} \Pr(S^{i,j} =s) \cdot \Pr(s
\text{ is accepted}) + \sum_{s \notin \Good_x} \Pr(S^{i,j} =s) \cdot
\Pr(s \text{ is accepted}) \\
& = & \sum_{s \in \Good_x} \Pr(S^{i,j} =s) \cdot \Pr(s
\text{ is accepted}) + 0 \\
& =& \sum_{s \in \Good_x} p(s) \cdot\frac{1}{2^{(k_a +1)/\delta^2}}
\cdot \frac{p^x(s)}{p(s)} = \frac{1}{2^{(k_a +1)/\delta^2}} \cdot
\Pr_{P_x}(s \in \Good_x) \geq \frac{1 - \delta}{2^{(k_a +1)/\delta^2}}.
\end{eqnarray*}
The last inequality follows from Eq.~(\ref{eq:goodx}). 
Therefore expected value of $J_i$ is $\frac{2^{(k_a +1)/\delta^2}}{1
- \delta}.$ Therefore, from concavity of the $\log$ function it
follows that the  expected communication from $\alice$ to communicate
$J_i$ to $\bob$ (using a prefix free binary encoding) is $O(\log
\frac{2^{(k_a +1)/\delta^2}}{1 - \delta})$.
This is true for every $i \in [K]$. Therefore for $x \in \Good_x$, expected
communication from $\alice$ is $O\left(\frac{\log \frac{1}{\delta}}{\delta^2}
\cdot (k_a + 1) \cdot 2^{O((k_b+1)/\delta^2)}\right)$. Therefore
overall expected communication from $\alice$ is $O\left(\frac{\log
\frac{1}{\delta}}{\delta^2}\cdot (k_a + 1) \cdot
2^{O((k_b+1)/\delta^2)}\right)$.

\smallskip
\noindent {\bf Expected error of $\cP'$:}
$\alice$ aborts the protocol when $x \notin \Good_X$, which happens with
probability at most $\delta$. Assume that $\alice$ does not abort. $\bob$
aborts the protocol when $y \notin
\Good_Y$, which happens with probability at most $\delta$. When 
$y \in \Good_Y$, using a similar calculation as above we can conclude
that $\bob$ accepts the $l$-th sample (for any $l \in [K]$), given that he
has rejected the samples before is at least $\frac{1-\delta}{2^{(k_b
+1)/\delta^2}}$. Therefore, $$ \Pr(\text{$\bob$ rejects all $K$ samples})
\leq \left(1 - \frac{1-\delta}{2^{(k_b +1)/\delta^2}}\right)^K \leq
\exp (-K \cdot \frac{1-\delta}{2^{(k_b +1)/\delta^2}}) = \delta. $$
Therefore, when $(x,y) \in \Good_X \times \Good_Y$, 
$$\Pr(\text{$\bob$ aborts given input of $\cP'$ is $(x,y)$}) \leq \delta.$$
We have the following claim.
\begin{claim} \label{claim:lowell}
Let $(x,y) \in \Good_X \times \Good_Y$ and $\bob$ does not abort. Then,
\begin{enumerate}
\item If $s \in \Good_x
\cap \Good_y$ then  $\Pr(\text{$\bob$ sets $M=s$}) =
\frac{p^{x,y}(s)}{\Pr(s \in \Good_x \cap \Good_y)}$. 
\item If $s \notin \Good_x
\cap \Good_y$ then $\Pr(\text{$\bob$ sets $M=s$}) = 0$.
\end{enumerate}
\end{claim}
We defer the proof of this claim to later. Let us now analyze the
expected error of the protocol $\cP'$ assuming
Claim~\ref{claim:lowell} to be true.  Let $\epsilon'_{x,y}$ stand for
error of $\cP'$ when input is $(x,y)$.  From above claim, if $(x,y)
\in \Good_X \times \Good_Y$ and $\bob$ does not abort, then the $\ell_1$
distance between the distribution of $M$ in $\cP'$ and $P_{x,y}$ is
$2(1-\Pr(s \in \Good_x \cap \Good_y)) \leq 2\delta$. Therefore if
$(x,y) \in \Good_X \times \Good_Y$ and $\bob$ does not abort, then
$\epsilon_{x,y}' \leq \epsilon_{x,y} +
\delta$, where $\epsilon_{x,y}$ is the error of $\cP$ on input
$(x,y)$. Therefore, for $(x,y)  \in \Good_X \times \Good_Y$,
$$\Pr(\cP' \text{ errs on input $(x,y)$  given $\bob$ does not abort})
\leq \epsilon_{x,y} + \delta.$$
This implies:
\begin{eqnarray*}
\lefteqn{\E_{(x,y) \in \Good_X \times \Good_Y}[\Pr(\cP' \text{ errs on input
$(x,y)$ given $\bob$ does not abort})]} \\
& \leq &  \E_{(x,y) \in \Good_X \times \Good_Y}[\epsilon_{x,y}] + \delta \\
& \leq & \frac{1}{\Pr((x,y) \in \Good_X \times \Good_Y)}\E_{(x,y) \in
\cX \times \cY}[\epsilon_{x,y}] + \delta \\
& \leq  & \frac{\epsilon}{1 - 2\delta} + \delta   
\end{eqnarray*}

\begin{eqnarray*}
\lefteqn{\text{Expected error of $\cP'$}} \\ 
& \leq & \Pr(x \notin \Good_X) + \Pr(y \notin \Good_Y) 
+ \Pr((x,y) \in \Good_X \times \Good_Y)\cdot
\E_{x \in \Good_X, y \in \Good_Y}[\epsilon_{x,y}']\\  
& \leq & \delta + \delta \\ 
&+& \Pr((x,y) \in \Good_X \times \Good_Y) \cdot
(\E_{x \in \Good_X, y \in \Good_Y}
[\Pr(\text{$\bob$ aborts given input of $\cP'$ is $(x,y)$})]  \\ 
&+& \E_{x \in \Good_X, y \in \Good_Y}[\Pr(\cP' \text{ errs given $x
\in \Good_X, y \in \Good_Y$ and $\bob$ does not abort})] ) \\
& \leq & 2\delta + \delta  + (\Pr((x,y) \in \Good_X \times \Good_Y)
\cdot \left(\frac{\epsilon}{\Pr((x,y) \in \Good_X \times \Good_Y)} +
\delta\right) \\
& \leq & 4 \delta + \epsilon. 
\end{eqnarray*}
\qed \end{proof}

We are now finally ready to describe the protocol $\tcP$. 

\begin{figure}[ht]
\begin{center}
\mybox{
\smallskip
{ \em Let $c$ be the expected communication from $\alice$ to $\bob$ in
protocol $\cP'$.}

\begin{description}

\item[Input:] $\alice$ gets $x \in X$ and $\bob$ gets $y \in Y$. 

\item[$\alice$:] She simulates protocol $\cP'$. If for some choice of
the public coins the bits needed to communicate all $J_i, i \in
[K]$ exceeds $c/\delta$, she aborts the protocol and sends a special
abort message to $\bob$ in constant bits.
\item[$\bob$:] In case he does not get abort message from $\alice$, he
proceeds as in protocol $\cP'$.
\end{description}
\smallskip
}
\end{center}
\caption{The final protocol $\tcP$} 
\label{fig:finalclass}
\end{figure}

Now it is clear that the communication of $\tcP$ is as claimed.  Also
it is easily noted that the expected error of $\tcP$ is at most
expected error of $\cP'$ plus $\delta$ which is $\epsilon + 5\delta =
\epsilon + \tilde{\delta}$ as claimed (since
$\delta = \frac{\tilde{\delta}}{5}$).

\begin{proofof}{Claim~\ref{claim:lowell}}
Let $l\in [K]$. Then conditioned on $\bob$ rejecting first $l-1$ samples,
for $s \in \Good_x \cap \Good_y$,
\begin{eqnarray*}
\lefteqn{\Pr(\text{$\bob$'s outputs $S^{l,J_l}$ and  $S^{l,j_l}=  s$})} \\
& = & \Pr(S^{l,j_l}= s) \cdot \Pr(\text{$\alice$ accepts $S^{l,j_l}$})
\cdot \Pr(\text{$\bob$ accepts $S^{l,j_l}$}) \\
& = & p(s) \cdot \frac{p^x(s)}{2^{(k_a+1)/\delta^2}p(s)} \cdot 
\frac{p^y(s)}{2^{(k_b+1)/\delta^2}p(s)} \\
& = &  \frac{p^{x,y}(s)}{2^{(k_a+k_b+2)/\delta^2}}.
\end{eqnarray*}
Therefore conditioned on $\bob$ rejecting first $l-1$ samples,
\begin{eqnarray*}
 \Pr(\text{$\bob$'s outputs $S^{l,j_l}$}) & = & 
\sum_{s \in S}\Pr(\text{$\bob$'s outputs $S^{l,j_l}$ and  $S^{l,j_l}=
s$}) \\
& = & \sum_{s \in \Good_x \cap \Good_y}\Pr(\text{$\bob$'s outputs $S^{l,j_l}$
and  $S^{l,j_l}= s$})  \\
& = & \frac{1}{2^{(k_a+k_b+2)/\delta^2}} \cdot \Pr(s \in \Good_x \cap \Good_y).
\end{eqnarray*}
Therefore, conditioned on $\bob$ rejecting first $l-1$ samples, for $s \in \Good_x
\cap \Good_y$,
\begin{eqnarray*}
 \Pr(\text{$\bob$'s outputs $s$ given $\bob$ outputs $S^{l,j_l}$}) & = & 
\frac{\Pr(\text{$\bob$'s outputs $S^{l,j_l}$ and  $S^{l,j_l}=  s$})
}{\Pr(\text{$\bob$'s outputs $S^{l,j_l}$})} \\
& = & \frac{p^{x,y}(s)}{\Pr(s \in \Good_x \cap \Good_y)}
\end{eqnarray*}
Clearly for $s \notin \Good_x \cap \Good_y, \Pr(\text{$\bob$'s outputs
$s$ given $\bob$ outputs $S^{l,j_l}$}) = 0$. Our claim now immediately follows.
\end{proofof}

As before we get the following corollaries from the above theorem.  
\begin{corollary}[Privacy tradeoff]
For any relation $f: \cX \times \cY \rightarrow \cZ$, 
$L(f, A, B) 2^{O(L (f, B, A))} \geq 
 \sR^{1, A \rightarrow B, [\;]} (f)$. Similarly,
$L (f, B, A) 2^{O(L (f, A, B))} \geq 
 \sR^{1, B \rightarrow A, [\;]} (f)$.
\end{corollary}

\begin{corollary}[Weak Direct Sum]
\label{corr:weakdirectclass} For any relation $f: \cX \times \cY \rightarrow \cZ$, 
$$\sR^{[\;]} (f^{\oplus m}) \quad \geq \quad m \cdot
\Omega(\log \sR^{1, [\;]} (f)).$$
\end{corollary}


\section{Entanglement Reduction}
\label{sec:nobboxred}
We will need the following geometric result. It is
similar to a result proved earlier
in \cite{jain:icalp03}.
\begin{lemma}
\label{lem:incompress}
Suppose $M$, $N$ are positive integers with 
$M = \Theta(N^{2/3} \log N)$.
Let the underlying Hilbert space be $\C^M$. 
There exist $16 N$ subspaces $V_{ij} \leq \C^M$,
$1 \leq i \leq N$, $1 \leq j \leq 16$, each of dimension 
$\frac{M}{16}$, such that if we define
$\Pi_{ij}$ to be the orthogonal projection onto $V_{ij}$ and
$\rho_{ij} \defeq \frac{16}{M} \cdot \Pi_{ij}$, 
then
\begin{enumerate}
\item $\forall i,j \, \Tr (\Pi_{ij}\rho_{ij}) = 1$.
\item $\forall i, j, i', j', \, i \neq i', \, 
       \Tr (\Pi_{ij}\rho_{i'j'}) < 1/4$.
\item $\forall i, j, j', \, j \neq j', \, 
       \Tr (\Pi_{ij}\rho_{ij'}) = 0$.
\item $\forall i, \, I_M = \sum_{j=1}^{16} \Pi_{ij}$, where $I_M$
      is the identity operator on $\C^M$.
\item For all subspaces $W$ of dimension at most $N^{1/6}$, 
      for all families of density matrices
      $\{\sigma_{ij}\}_{i \in [N], 1 \leq j \leq 16}$,
      $\sigma_{ij}$ supported  in $W$,
      \[
      |\{i: \exists j, \; 1 \leq j \leq 16,  \;
            \Tr (\Pi_{ij} \sigma_{ij}) > 9/16\}| \leq N/4. 
      \]
\end{enumerate}
\end{lemma}
\begin{proof}({\bf Sketch})
The proof follows by combining the proofs of Thm.~5 and
Lem.~7 of \cite{jain:icalp03}. We skip a full proof for brevity.
\qed \end{proof}

We shall also need the following easy proposition.
\begin{proposition}
\label{prop:schmidt}
Let $\ket{\phi}_{AB}$ be a bipartite pure quantum state.
Define $e \defeq E(\ket{\phi})$. Then there is a bipartite
pure quantum state $\ket{\phi'}_{AB}$ having Schmidt rank at
most $2^{100 e}$ such that 
$\trnorm{\ketbra{\phi} - \ketbra{\phi'}} \leq 1/20$.
\end{proposition}
\begin{proof}
Let 
$\ket{\phi}_{AB} = \sum_i \sqrt{\lambda_i} \ket{a_i}_A \ket{b_i}_B$
be the Schmidt decomposition of $\ket{\phi}$, $\lambda_i \geq 0$,
$\sum_i \lambda_i = 1$. Define a set
$\Good \defeq \{i: \lambda_i \geq 2^{-100 e}\}$.
Since $e = -\sum_i \lambda_i \log \lambda_i$, by Markov's inequality
$\sum_{i \in \Good} \lambda_i \geq 99/100$. Define the bipartite
pure state 
$\ket{\phi'}_{AB} \defeq 
 \sum_{i \in \Good} \sqrt{\lambda_i} \ket{a_i}_A \ket{b_i}_B$
normalized. The Schmidt rank of $\ket{\phi'}_{AB}$ is at most
$2^{100 e}$ and 
$\trnorm{\ketbra{\phi} - \ketbra{\phi'}} \leq 1/20$.
\qed \end{proof}

We are now ready to prove our impossibility result about black-box
reduction of prior entanglement.
\begin{theorem}[No black-box red. of prior entan.]
\label{thm:nobboxred}
Let $\EQ_n$ denote the Equality function on $n$-bit strings.
There exists a one-round quantum protocol $\cP$ for $\EQ_n$ with 
$\frac{2n}{3} + \log n + \Theta(1)$ EPR pairs of prior 
entanglement and communicating $4$ bits
such that, there is no
similar protocol $\cP'$ that starts with a prior entangled
state $\ket{\phi}$, $E(\ket{\phi}) \leq \frac{n}{600}$.
\end{theorem}
\begin{proof}
We use the notation of Lem.~\ref{lem:incompress} with
$M \defeq 2^m$ and $N \defeq 2^n$.
Let $0 \leq i \leq 2^n - 1$ i.e. $i \in \{0, 1\}^n$. 
Choose $m = \frac{2n}{3} + \log n + \Theta(1)$.
Let $\cP$ be a one-round protocol with $m$ EPR pairs of prior
entanglement.  In $\cP$, on input $i$ $\alice$ measures her EPR
halves according to the von-Neumann measurement 
$\{\Pi_j\}_{1 \leq j \leq 16}$ and sends the result $j$ as a $4$-bit  
classical message to $\bob$. The state of $\bob$'s EPR halves now becomes
$\rho_{ij}$. On input $i'$ and  
message $j'$, $\bob$ performs the two-outcome measurement 
$\{\Pi_{i'j'}, I_M - \Pi_{i'j'}\}$ on his EPR halves.  
Therefore in $\cP$, $\bob$ outputs $1$ with probability $1$
if $i' = i$ and with probability at most $1/4$ if $i' \neq i$.
Thus, $\cP$ is a protocol for $\EQ_n$.

Suppose there exists a protocol $\cP'$ similar to $\cP$ 
that starts with an input independent shared state 
$\ket{\phi'}_{AB}$ on
$m + m$ qubits. Suppose $E(\ket{\phi}) \leq n/10$.
By Proposition~\ref{prop:schmidt},
there is a bipartite pure state $\ket{\phi''}_{AB}$ on
$m + m$ qubits having Schmidt rank at most $2^{n/6}$ such
that $\trnorm{\ketbra{\phi'} - \ketbra{\phi''}} \leq 1/20$.
Consider the protocol $\cP''$ similar to $\cP'$ starting
with $\ket{\phi''}_{AB}$ as prior entanglement.
Since $\cP''$ is similar to $\cP'$, it is also
a one-round protocol with $4$ classical bits of communication. 
Let $\sigma_{ij}$ be the state of $\bob$'s 
share of prior entanglement qubits after the first round of 
communication from $\alice$ when
$\alice$'s input is $i$ and her message is $j$. Since the Schmidt 
rank of $\ket{\phi''}$ is at most $2^{n/6}$, the $\sigma_{ij}$,
$0 \leq i \leq 2^n - 1$, $1 \leq j \leq 16$
have support in a $2^{n/6}$-dimensional space.
Let $p_{ij}$ be the probability with which $\alice$ sends message $j$
when her input is $i$. 
It follows that for all $i$,
$\sum_{j=1}^{16} p_j \Tr M_{ij} \sigma_{ij} \geq 
 \frac{3}{4} - \frac{1}{20} - \frac{1}{20} = \frac{13}{20}$.
This implies that for all $i$ there exists a $j$,
$1 \leq j \leq 16$, such that
$\Tr M_{ij} \sigma_{ij} \geq 13/20 > 9/16$.
From Lem.~\ref{lem:incompress} 
this is not possible, and hence no such protocol $\cP'$ exists.
\qed \end{proof}

\section{Exact Remote State Preparation}
\label{sec:rsp}

\begin{proofof}{Thm.~\ref{thm:ersp}}
We start with the following lemma which may be of independent interest.

\begin{lemma}
\label{lem:sitinsidepure}
Let $\rho \defeq \ketbra{\phi} \in \cH$ be a pure state and $\sigma
\in \cH$ be any positive definite matrix. Then the
maximum value of $k$ such that, $ \sigma - k \rho \geq 0$, is
$(\bra{\phi}\sigma^{-1}\ket{\phi})^{-1}$.
\end{lemma}
\begin{proof}
First we show that, $ \bra{\phi}\sigma^{-1}\ket{\phi} \sigma - \rho
\geq 0$.  Let $\ket{v} \in \cH$. Let $\ket{w_1} \defeq
\sigma^{-1/2}\ket{\phi}$ and $\ket{w_2}
\defeq \sigma^{1/2}\ket{v}$. Now Cauchy-Schwartz inequality implies,
\begin{eqnarray*}
\braket{w_1}{w_1} \braket{w_2}{w_2} & \geq & |\braket{w_1}{w_2}|^2
\\
\Rightarrow \bra{\phi}\sigma^{-1}\ket{\phi} \bra{v}\sigma \ket{v} &
\geq & |\braket{\phi}{v}|^2 \\
\Rightarrow   \bra{v}\bra{\phi}\sigma^{-1}\ket{\phi} \sigma \ket{v}
& \geq & \braket{v}{\phi} \braket{\phi}{v}  \\
\Rightarrow  \bra{v}(\bra{\phi}\sigma^{-1}\ket{\phi} \sigma  -
\ketbra{\phi} )\ket{v} & \geq & 0
\end{eqnarray*}
Now since above is true for every $\ket{v} \in \cH$ we have that
$\bra{\phi}\sigma^{-1}\ket{\phi} \sigma  -
\ketbra{\phi} \geq 0$. \\
Next we show that if $k > (\bra{\phi}\sigma^{-1}\ket{\phi})^{-1}$ then
$\sigma - \ketbra{\phi}$ is not positive semi-definite. For this let
$\ket{v} \defeq \sigma^{-1}\ket{\phi}$, and in this case $\ket{w_1} =
\ket{w_2}$. Now since $\sigma \geq 0$ and $k >
(\bra{\phi}\sigma^{-1}\ket{\phi})^{-1}$ we have,
\begin{eqnarray*} 
\bra{v}(k^{-1} \sigma  - \ketbra{\phi})\ket{v} & < &  \bra{v}(
\bra{\phi}\sigma^{-1}\ket{\phi} \sigma  - \ketbra{\phi})\ket{v} \\
&= & \bra{\phi}\sigma^{-1}\ket{\phi} \bra{v}\sigma \ket{v} - |\braket{\phi}{v}|^2 \\
& = & \braket{w_1}{w_1} \braket{w_2}{w_2} - |\braket{w_1}{w_2}|^2 \\ 
& =&  0
\end{eqnarray*}
Hence $k^{-1} \sigma  - \ketbra{\phi}$ is not positive semi-definite.
\qed \end{proof}

Let $\rho \defeq \ketbra{\phi}$, $\sigma$ be some full rank state
and let $k = (\bra{\phi}\sigma^{-1}\ket{\phi})^{-1}$.  Let $\rho'
\defeq \sigma - \bra{\phi}\sigma^{-1}\ket{\phi})^{-1}
\rho$. Lem.~\ref{lem:sitinsidepure} implies $\rho' \geq 0$. Let
$\cK$ be a Hilbert space with $dim(\cK) = dim(\cH)$. Let
$\ket{\theta}  \in \cK \otimes \cH$ be some purification of $\rho'$
and $\ket{\bar{0}}$ be a fixed vector in $\cK$. We now define,
$$\ket{\psi}_{\rho} \defeq
\sqrt{k}
\ket{1}\ket{\bar{0}}\ket{\phi} + \sqrt{1 - k} \ket{0} \ket{\theta} $$
We note that the marginal of $\ket{\psi}_{\rho}$ in $\cH$ is
$\sigma$.

We have the following lemma due to Jozsa and
Uhlmann~\cite{Jozsa94,Uhlmann76}.

\begin{lemma}[Local transition]
\label{lem:loctrans}
Let $\rho$ be a quantum state in $\cH$. Let $\ket{\phi_1}$ and $
\ket{\phi_2}$ be two purification of $\rho$ in $\cK \otimes \cH$. 
There is a local unitary transformation $U$ acting on $\cK$ such that
$(U \otimes I) \ket{\phi_1} = \ket{\phi_2}$. 
\end{lemma}

Now consider the following protocol $\cP$:
\begin{enumerate}
\item $\alice$ and $\bob$ start with several copies of a fixed pure
state $\ket{\psi}$ such that marginal on $\bob$'s side in $\ket{\psi}$ is $\sigma$.
\item On getting $x$,  $\alice$ transforms using a local unitary the first copy
of $\ket{\psi}$ to $\ket{\psi}_{\rho_x}$. This can be done using
Lemma~\ref{lem:loctrans}, since the
marginal on $\bob$'s side in both $\ket{\psi}$ and $\ket{\psi}_{\rho_x}$ is
$\sigma$. She then measures the first qubit. 
\item She keeps doing this to successive copies of $\ket{\psi}$
until she gets the first 1 on measurement. She communicates to $\bob$ the first
occurrence of 1.
\end{enumerate}

From the definition of $\ket{\psi}_{\rho_x}$, we note that in the
copy in which $\alice$ gets 1, $\bob$ ends up with $\rho_x$. Also,
(from concavity of the $\log$ function) it
can be verified that, using a prefix-free encoding of integers that
requires $\log n + 2 \log \log n$ bits to encode the integer $n$, the expected
communication of $\alice$ is bounded by $\log (\Tr \sigma^{-1} \rho_x)
+ 2\log \log (\Tr \sigma^{-1} \rho_x)$.
Hence our theorem.
\end{proofof}

\paragraph{Remarks:} 
\begin{enumerate}
\item For any fixed state $\sigma$ of full rank, from the above
proof, we get a protocol $\cP_{\sigma}$ such that given the
description of any pure state $\rho$ to $\alice$, she ends up creating
$\rho$ with $\bob$ with communication $\log (\Tr \sigma^{-1} \rho)$.
\item 
We note that when $\rho_x \defeq \ketbra{\phi_x}$ then from concavity of
$\log$ function we have, $S(\rho_x ||
\sigma) = \bra{\phi_x} \log \sigma \ket{\phi} \leq \log
\bra{\phi}\sigma^{-1}\ket{\phi_x}$. Therefore the approach that we take
here, which is analogous to the {\em rejection 
sampling} approach of~\cite{HarshaJMR07}, does not help
us in getting the communication down to $S(\rho_x || \sigma)$ which
happens in~\cite{HarshaJMR07} for a similar problem in the classical setting. 
\item It is
open as to whether the communication could be brought down to
$S(\rho_x ||\sigma)$. Also the case when $\rho_x$ is not necessarily a pure
state is open.
\item The inexact version of this problem was considered in~\cite{jain:remote}
where some fidelity loss in generating $\rho_x$ was allowed. There using
the substate theorem, the task was accomplished with communication
$S(\rho_x || \sigma)/\epsilon$ at the end of which $\bob$ got a state
$\rho_x'$ which was $\epsilon$ close in trace distance to $\rho_x$ (not
necessarily pure). 
\end{enumerate}

\subsection*{Acknowledgment} We thank the referees for their
comments and suggestions. We are grateful to Harold Ollivier for his
comments on the proof of Thm.~\ref{thm:nobboxred}.


\end{document}